\documentclass[ejs]{imsart}

\usepackage{amsmath,amsfonts,amsthm,color,mathrsfs,amssymb}
\usepackage{natbib}
\usepackage[utf8]{inputenc}
\usepackage[T1]{fontenc} 
\usepackage{graphicx}

\newtheorem{theorem}{Theorem}
\newtheorem{proposition}{Proposition}
\newtheorem{corollary}{Corollary}

\newtheorem{lemma}{Lemma}
\newtheorem{remark}{Remark}

\DeclareMathOperator*{\argmax}{arg\,max}

\begin{document}

\begin{frontmatter}

\title{Bayesian inference for multivariate extreme value distributions}
\runtitle{Bayesian inference for multivariate extremes}


\author{\fnms{Cl\'ement} \snm{Dombry}\ead[label=e1]{clement.dombry@univ-fcomte.fr}}
\address{Univ. Bourgogne Franche--Comt\'e\\
Laboratoire de Math\'ematiques de Besan\c{c}on\\
UMR CNRS 6623, 16 Route de Gray\\
25030 Besan\c{c}on cedex, France\\
\printead{e1}}

\author{\fnms{Sebastian} \snm{Engelke}\ead[label=e2]{sebastian.engelke@epfl.ch}}
\address{\'Ecole Polytechnique F\'ed\'erale de Lausanne\\
  EPFL-FSB-MATHAA-STAT, Station 8\\
  1015 Lausanne, Switzerland\\
\printead{e2}}

\author{\fnms{Marco} \snm{Oesting}\ead[label=e3]{oesting@mathematik.uni-siegen.de}}
\address{Universit\"at Siegen\\
  Department Mathematik\\
  Walter-Flex-Str.~3\\
  57068 Siegen, Germany\\
\printead{e3}}

\runauthor{Dombry et al.}

\begin{abstract}
  Statistical modeling of multivariate and spatial extreme events has attracted
  broad attention in various areas of science. Max-stable distributions and 
  processes are the natural class of models for this purpose, and many 
  parametric families have been developed and successfully applied. Due to 
  complicated likelihoods, the efficient statistical inference is still an 
  active area of research, and usually composite likelihood methods based on 
  bivariate densities only are used. Thibaud et al. [2016, 
  \emph{Ann.~Appl.~Stat.}, to appear] use a Bayesian approach to fit a 
  Brown--Resnick process to extreme temperatures. In this paper, we extend this 
  idea to a methodology that is applicable to general max-stable distributions 
  and that uses full likelihoods. We further provide simple conditions for the
  asymptotic normality of the median of the posterior distribution and verify
  them for the commonly used models in multivariate and spatial extreme value
  statistics. A simulation study shows that this point estimator is 
  considerably more efficient than the composite likelihood estimator in a 
  frequentist framework. From a Bayesian perspective, our approach opens the 
  way for new techniques such as Bayesian model comparison in multivariate and 
  spatial extremes. 
\end{abstract}



\end{frontmatter}

\section{Introduction} \label{sec:intro}

Extremes and the impacts of rare events have been brought into public focus in 
the context of climate change or financial crises. The temporal or spatial 
concurrence of several such events has often shown to be most catastrophic. 
Arising naturally as limits of rescaled componentwise maxima of random vectors,
max-stable distributions are frequently used to describe this joint behavior of
extremes. The generalization to continuous domains gives rise to max-stable 
processes that have become popular models in spatial extreme value statistics
\citep[e.g.,][]{dav2012}, and are applied in various fields such as meteorology
\citep{bui2008, eng2014, ein2016} and hydrology \citep{asa2015}.

For a $k$-dimensional max-stable random vector $Z = (Z_1,\dots,Z_k)$ with unit 
Fr\'echet margins, there exists an exponent function $V$ describing the 
dependence between the components of $Z$ such that $\mathbb{P}[Z\leq z] = 
\exp\{-V(z)\}$, $z\in (0,\infty)^k$. Many parametric models 
$\{F_\theta,\ \theta \in \Theta\}$ for the distribution function of $Z$ have 
been proposed \citep[cf.][]{sch2002, bol2007, kab2009, opi2013},
but likelihood-based inference remains challenging. The main reason is the lack
of simple forms of the likelihood $L(z;\theta)$ in these models, which, by 
Fa\'a di Bruno's formula, is given by
\begin{equation} \label{eq:llh_Z}
  L(z;\theta) = \sum_{\tau\in \mathscr{P}_k} L(z,\tau;\theta)
              = \sum_{\tau\in \mathscr{P}_k} \exp\{-V(z)\} \prod_{j=1}^{|\tau|} \{-\partial_{\tau_j}V(z)\},
\end{equation}
where $\mathscr{P}_k$ is the set of all partitions $\tau = \{\tau_1,\ldots,
\tau_{|\tau|}\}$ of $\{1,\ldots,k\}$ and $\partial_{\tau_j}V(\cdot;\theta)$ 
denotes the partial derivative of the exponent function $V=V_\theta$ of 
$F_\theta$ with respect to the variables $z_i$, $i\in \tau_j$. The fact that 
the cardinality of $\mathscr{P}_k$ is the $k$th Bell number that grows 
super-exponentially in the dimension $k$ inhibits the use of the maximum 
likelihood methods based on $L(z;\theta)$ in \eqref{eq:llh_Z}.

The most common way to avoid this problem is to maximize the composite pairwise
likelihood that relies only on the information in bivariate sub-vectors of $Z$ 
\citep{PRS10}. Apart from the fact that this likelihood is misspecified, there 
might also be considerable losses in efficiency by using the composition of
bivariate likelihoods instead of the full likelihood $L(z;\theta)$. To avoid 
this efficiency loss, higher order composite likelihood has been considered
\citep{GMS11,HD13,CHG16}.

In practice, inference is either based on threshold exceedances \citep{eng2012,
wadsworth-tawn14, thi2015b} or the block maxima method. For the latter, the 
data, typically a multivariate time series, is split into disjoint blocks and a
max-stable distribution is fitted to the componentwise maxima within each 
block. To increase the efficiency, not only the block maxima but additional 
information from the time series can be exploited. The componentwise occurrence
times of the maxima within each block lead to a partition $\tau$ of 
$\{1,\ldots,k\}$ with indices belonging to the same subset if and only if the 
maxima in this component occurred at the same time. The knowledge of this 
partition makes inference much easier, as a single summand $L(z,\tau;\theta)$
in the full likelihood $L(z;\theta)$ given in \eqref{eq:llh_Z} corresponds to 
the likelihood contribution of the specific partition $\tau$. This joint 
likelihood $L(z,\tau;\theta)$ was introduced in \citet{ST05} and is tractable 
for many extreme value models and, consequently, can be used for inference if 
occurrence times are available. 

In real data applications, however, the distribution of the block maxima is
only approximated by a max-stable distribution and the distribution of the 
observed partitions of occurrence times are only approximations to the limiting
distribution (as the block size tends to infinity) given by the likelihood 
$L(z,\tau;\theta)$. This approximation introduces a significant bias in the 
Stephenson--Tawn estimator and a bias correction has been proposed in 
\citet{wadsworth14}.

In many cases, only observations $z^{(1)},\ldots,z^{(N)} \in\mathbb R^k$ of the
random max-stable vector $Z$ are available, but there is no information about 
the corresponding partitions $\tau^{(1)},\dots,\tau^{(N)}$. In this case, the 
Stephenson--Tawn likelihood cannot be used since the partition information is
missing. In the context of conditional simulation of max-stable processes, 
\citet{DEMR13} proposed a Gibbs sampler to obtain conditional samples of 
$\tau^{(l)}$ given the observation $z^{(l)}$, $l=1,\ldots,N$. \cite{thi2015} 
use this approach to treat the missing partitions as latent variables in a 
Bayesian framework to estimate the parameters of a Brown--Resnick model 
\citep[cf.,][]{kab2009} for extreme temperature. They obtain samples from the
posterior distribution
\begin{equation}\label{posterior}
  L\left(\theta,\{\tau^{(l)}\}_{l=1}^N \mid \{z^{(l)}\}_{l=1}^N\right) \propto \pi_\theta(\theta) \prod_{l=1}^N L(z^{(l)},\tau^{(l)};\theta),
\end{equation}
via a Markov chain Monte Carlo algorithm, where $\pi_\theta$ is the prior 
distribution on $\Theta$.

In this paper we extend the Bayesian approach to general max-stable 
distributions and provide various examples of parametric models $F_\theta$ 
where it can be applied. The first focus is to study the statistical 
efficiency of the point estimators obtained as the median of the posterior 
distribution \eqref{posterior}. This frequentist perspective allows to compare
the efficiency of the Bayesian estimator that uses the full likelihoods to 
other frequentist estimators. A simulation study shows a substantial 
improvement of the estimation error when using full likelihoods rather than 
the commonly used pairwise likelihood estimator of \citet{PRS10}.

From the Bayesian perspective, this approach opens up many new possibilities
for Bayesian techniques in multivariate extreme value statistics. Besides 
readily available credible intervals, we discuss how Bayesian model comparison
can be implemented. Thanks to the full, well-specified likelihoods in our 
approach, no adjustment of the posterior distribution as in the composite 
pairwise likelihood methods \citep{rib2012} is required.

In Section \ref{sec:metho} we provide some background on max-stable 
distributions and their likelihoods, and we present the general methodology for
the Bayesian full-likelihood approach. Section \ref{sec:asymp} develops an 
asymptotic theory for the resulting estimator. We show in Section 
\ref{sec:examples} that our method and the asymptotic theory are applicable for 
the popular models in multivariate and spatial extremes, including the 
 Brown--Resnick and extremal-$t$ processes. The simulation studies
in Section \ref{sec:simu} quantify the finite-sample efficiency gains of the 
Bayesian approach when used as a frequentist point estimator of the extremal
dependence parameters. Interestingly, this advantage persists when the 
dependence is a nuisance parameter and one is only interested in estimating
marginal parameters (Section \ref{est_margins}). The posterior distribution and 
genuinely Bayesian techniques are studied in Section \ref{sec:bayes}, with a 
focus on Bayesian model comparison. Section 
\ref{sec:disc} concludes the paper with a discussion on computational aspects.

\section{Methodology} \label{sec:metho}

In Section \ref{subsec:maxstab} we review some facts on max-stable distributions 
and their likelihoods. We describe the general setup 
of our approach and review the Markov chain Monte Carlo algorithm from
\citet{thi2015} and the Gibbs sampler from \citet{DEMR13} in Section 
\ref{subsec:bayes_setup}. 

\subsection{Max-stable distributions, partitions and joint likelihoods}\label{subsec:maxstab}

Let us assume from now on that the max-stable vector $Z$ belongs to a 
parametric family $\{F_\theta, \theta \in \Theta\}$, where $\Theta \subset 
\mathbb R^p$ is the parameter space, and that it admits a density $f_\theta$.
The exponent function of $F_\theta$ is $V_\theta(z) = -\log F_\theta(z)$. If 
there is no confusion we might omit the dependence on $\theta$ for simplicity.

Recall that if $Z$ has standard Fr\'echet margins, it can be represented as the
componentwise maximum 
\begin{equation} \label{spec-repr}
 Z_i = \max_{j\in\mathbb N} \psi^{(j)}_i, \quad i=1,\ldots,k,
\end{equation}
where $\{\psi^{(j)}: j \in \mathbb N\}$ is a Poisson point process on 
$E = [0,\infty)^k\setminus \{0\}$ with intensity measure $\Lambda$ such that
$\Lambda( E\setminus [0,z])=V(z)$. For more details and an exact simulation 
method of $Z$ via this representation, we refer to \citet{dom2016}.

Analogously to the occurrence times in case of block maxima, the Poisson point 
process induces a random limit partition $T$ of the index set $\{1,\ldots,k\}$, 
where two indices $i_1 \neq i_2$ belong to the same subset if and only if 
$Z_{i_1} = \psi^{(j)}_{i_1}$ and $Z_{i_2} = \psi^{(j)}_{i_2}$ for the same 
$j \in \mathbb N$ \citep{DEM13}. The joint likelihood of the max-stable vector
$Z$ and the limit partition $T$ under the model $F_\theta$ satisfies
\begin{equation} \label{llh_ST}
 L(z,\tau; \theta) = \exp\{-V(z)\} \prod_{j=1}^{|\tau|} \{-\partial_{\tau_j}V(z)\}, \quad z \in (0,\infty)^k, \ \tau \in \mathscr{P}_k,
\end{equation}
 and it equals the likelihood introduced in \citet{ST05}. This fact provides 
another interpretation of Equation \eqref{eq:llh_Z}, namely that the likelihood
of $Z$ is the integrated joint likelihood of $Z$ and $T$.

In \cite{deo2017} it has been shown that the existence of a density for the
simple max-stable random vector $Z$ with exponent measure $\Lambda$ is 
equivalent to the existence of a density $\lambda_I$ for the restrictions of
$\Lambda$ to the different faces $E_I\subset E$ defined by
$$ E_I=\{z\in E;\ z_I>0 \mbox{ for } i \in I \mbox{ and } z_{I^c}=0 \mbox{ for } i \notin I\}, \quad \emptyset\neq I\subset \{1,\ldots,k\},$$
that is,
$$ \Lambda(A)=\sum_{\emptyset\neq I\subset\{1,\ldots,k\}} \int_{\{z_I: \, z \in A \cap E_I\}} \lambda_I(z_I)\mu_I(\mathrm{d}z_I),$$
Thus, the Stephenson--Tawn likelihood $L(z,\tau;\theta)$ can be rewritten as
\begin{equation}\label{eq:def_L_tau_frechet}
  L(z, \tau;\theta)=\exp\{-V(z)\}\prod_{j=1}^\ell \omega(\tau_j,z),
\end{equation}
where
\begin{equation}\label{weights}
 \omega(\tau_j, z) = \sum_{\tau_j \subset I \subset\{1,\ldots,k\}} \int_{(0,z_{\tau_j^c \cap I})} \lambda_I(z_{\tau_j},u_j) \mathrm{d}u_j,
\end{equation}
and $\tau_1,\ldots,\tau_\ell$ denote the $\ell = |\tau|$ different blocks of 
the partition $\tau$, and $z_{\tau_j}$ and $z_{\tau_j^c}$ are the restrictions
of $z$ to $\tau_j$ and $\tau_j^c= \{1, \ldots, k\} \setminus \tau_j$, 
respectively. 

Equation \eqref{eq:def_L_tau_frechet} provides a formula for the joint 
likelihood of max-stable distributions with unit Fr\'echet margins and its 
partition. From this we can deduce a formula for the joint likelihood of a
general max-stable distribution that admits a density. More precisely, let
$\bar Z$ be a $k$-dimensional max-stable random vector whose $i$th component,
$i=1,\dots,k$, has a generalized extreme value distribution with parameters 
$(\mu_i,\sigma_i,\xi_i) \in \mathbb R \times (0,\infty) \times \mathbb R$, 
that is,
\begin{align*}
  \mathbb P(\bar Z_i \leq z_i) = \exp\left\{- \left(1 + \xi_i \frac{z_i - \mu_i}{\sigma_i}\right)_+^{-1/\xi_i}\right\},   \quad z_i \in \mathbb R.
\end{align*}
Then, $U_i(\bar Z_i)$ has unit Fr\'echet distribution where $U_i$ denotes the 
marginal transformation 
\begin{align*}
  U_i(x)  = \left(1 + \xi_i \frac{x - \mu_i}{\sigma_i}\right)^{1/\xi_i}, \quad 1 + \xi_i \frac{x - \mu_i}{\sigma_i} > 0.
\end{align*} 
For a vector $z_I = (z_i)_{i\in I}$ with $I \subset \{1,\dots, k\}$, we define 
$U(z_I) = (U_i(z_i))_{i\in I}$, such that 
$\mathbb P(\bar Z \leq z) = \exp\left\{-V[U(z)]\right\}$,
where $V$ is the exponent measure of the normalized max-stable distribution
$U(\bar Z)$. Consequently, the joint density of the general max-stable vector
$\bar Z$ and the limit partition $T$ is 
\begin{align} \label{eq:def_L_tau_general} 
  L(z, \tau; \theta) = \exp\left\{ - V[U(z)]\right\} \cdot  \left(\prod_{j=1}^\ell \omega(\tau_j,U(z))\right) \cdot
 \left(\prod_{i = 1}^k \frac{1}{\sigma_i} U_i(z_i)^{1 -\xi_i}\right),  
\end{align}
for $z\in (0,\infty)^k$ such that $1 + \xi_i (z_i-\mu_i)/\sigma_i>0,\ i=1,\ldots,k$ and $\tau = \{\tau_1,\ldots,\tau_\ell\} \in \mathscr{P}_k$.

\subsection{Bayesian inference and Markov chain Monte Carlo}
\label{subsec:bayes_setup}

Extreme value statistics is concerned with the estimation and uncertainty 
quantification of the parameter vector $\theta \in \Theta$. Here, $\theta$ 
might include both marginal and dependence parameters of the max-stable model. 
In a Bayesian setup we introduce a prior $\pi_{\theta}(\theta)$ on the
parameter space $\Theta$. Given independent data $z^{(1)},\dots, z^{(N)} \in 
\mathbb R^k$ from the max-stable distribution $Z\sim F_\theta$, we are
interested in the posterior distribution of the parameter $\theta$ conditional 
on the data. As explained in Section \ref{sec:intro}, the complex structure of
the full likelihood $L(\{z^{(l)}\};\theta) = \prod_{l=1}^N L(z^{(l)};\theta)$ 
prevents a direct assessment of the posterior distribution, which is
proportional to the product of $L(\{z^{(l)}\};\theta)$ and the prior density 
$\pi_{\theta}(\theta)$. Instead, we introduce the corresponding limit 
partitions $T^{(1)},\dots, T^{(N)}$ as latent variables and sample from the 
joint distribution of $(\theta, T^{(1)}, \dots, T^{(N)})$ conditional on the 
data $z^{(1)}, \dots, z^{(N)}$, which is given in Equation \eqref{posterior}.

It is customary to use Monte Carlo Markov Chain methods to sample from a target 
distribution which is known up to a multiplicative constant only. The aim is to 
construct a Markov chain which possesses the target distribution as stationary
distribution and has good mixing properties. To this end, in each step of the 
Markov chain, the parameter vector $\theta$ and the partitions $T^{(1)}, \dots, 
T^{(N)}$ are updated separately by the Metropolis--Hastings algorithm and a 
Gibbs sampler, respectively \citep[cf.,][]{thi2015}.

For fixed partitions $T^{(l)}=\tau^{(l)}$, $l=1,\dots, N$, and the current 
state $\theta$ for the parameter vector, we propose a new state $\theta^\ast$ 
according to a probability density $q(\theta,\cdot)$ which satisfies 
$q(\theta_1,\theta_2)>0$ if and only if $q(\theta_2,\theta_1)>0$ for $\theta_1,
\theta_2 \in \Theta$. The proposal is accepted, that is, $\theta$ is updated to
$\theta^\ast$, with probability
\begin{equation}
 a(\theta,\theta^\ast)=\min\left\{\frac{\prod_{l=1}^N L(z^{(l)},\tau^{(l)};\theta^\ast) \pi_{\theta}(\theta^\ast) q(\theta^\ast,\theta)}
                                       {\prod_{l=1}^N L(z^{(l)},\tau^{(l)};\theta)      \pi_{\theta}(\theta)      q(\theta,\theta^\ast)}, 1\right\}
\end{equation}
where $L(z,\tau;\theta)$ is given by \eqref{eq:def_L_tau_general}. In general,
there are various ways of choosing an appropriate proposal density $q$. For 
instance, it might be advisable to update the vector $\theta$ component by 
component. It has to be ensured that any state $\theta_2$ with positive 
posterior density can be reached from any other state $\theta_1$ with positive 
posterior density in a finite number of steps, that is, that the Markov chain 
is irreducible. The convergence of the Markov chain to its stationary 
distribution \eqref{posterior} is then guaranteed. Note that the framework 
described above enables estimation of marginal and dependence parameters 
simultaneously. In particular, it allows for response surface methodology such 
as (log-)linear models for the marginal parameters.

For a fixed parameter vector $\theta \in \Theta$ we use the Gibbs sampler in
\citet{DEMR13} to update the current states of the partitions $\tau^{(1)},\dots, \tau^{(N)} \in\mathcal{P}_k$
conditional on the data $z^{(1)},\dots, z^{(N)}$. Thanks to independence, for 
each $l=1,\dots, N$, we can update $\tau = \tau^{(l)}$ conditional on 
$z = z^{(l)}$ separately, where the conditional distribution is 
\begin{align} \label{eq:L-tau}
 L(\tau \mid z; \theta) &={} \frac{L(z, \tau; \theta)}{\sum_{\tau' \in \mathscr{P}_k} L(z, \tau'; \theta)} = \frac{1}{C_{z}} \prod_{j=1}^\ell \omega\{\tau_j, U(z)\},
\end{align}
with $C_{z}$  the normalization constant
\begin{equation*}
C_{z}=\sum_{\tau\in\mathcal{P}_k}\prod_{j=1}^\ell \omega\{\tau_j,U(z)\}.
\end{equation*}

For $i\in\{1,\ldots,k\}$, let $\tau_{-i}$ be the restriction of $\tau$ to the 
set $\{1,\ldots,k\} \setminus \{i\}$. As usual with Gibbs samplers, our goal
is to simulate from 
\begin{equation} \label{eq:fullCondDist}
  \mathbb{P}_\theta\left(T = \cdot \mid T_{-i} = \tau_{-i}, \, Z=z\right),
\end{equation}
where $\tau$ is the current state of the Markov chain and $\mathbb{P}_\theta$ 
denotes the probability under the assumption that $Z$ follows the law 
$F_\theta$. It is easy to see that the number of possible updates according to 
\eqref{eq:fullCondDist} is always less than $k$, so that a combinatorial 
explosion is avoided. Indeed, the index $i$ can be reallocated to any of the 
components of $\tau_{-i}$ or to a new component with a single point: the number
of possible updates $\tau^\ast \in\mathcal{P}_k$ such that $\tau^\ast_{-i} = 
\tau_{-i}$ equals $\ell$ is $\{i\}$ if a partitioning set of $\tau$, and
$\ell + 1$ otherwise.

The distribution~\eqref{eq:fullCondDist} has nice properties. From
\eqref{eq:L-tau}, we obtain that
\begin{equation}\label{eq:condDistPart}
  \mathbb{P}_\theta(T = \tau^\ast \mid T_{-i} = \tau_{-i}, Z=z) =\frac{L(z,\tau^\ast)}{{\displaystyle
      \sum_{\tau' \in \mathcal{P}_k} L(z, \tau') 1_{\{\tau'_{-i} = \tau_{-i}\}}}} \propto
  \frac{\prod_{j=1}^{|\tau^\ast|} w\{\tau^\ast_j, U(z)\}}{\prod_{j=1}^{|\tau|}  w\{\tau_j,U(z)\}}.  
\end{equation}
for all $\tau^\ast \in \mathscr{P}_k$ with $\tau^\ast_{-i} = \tau_{-i}$. Since
$\tau$ and $\tau^\ast$ share many components, all the factors in the right-hand
side of~\eqref{eq:condDistPart} cancel out except at most four of them. This 
makes the Gibbs sampler particularly convenient.

We suggest a random scan implementation of the Gibbs sampler, meaning that one 
iteration of the  Gibbs sampler selects randomly an element $i \in \{1,\ldots,
k\}$ and then updates the current state $\tau$ according to the proposal 
distribution \eqref{eq:fullCondDist}. For the sake of simplicity, we use the
uniform random scan, i.e., $i$ is selected according to the uniform 
distribution on $\{1,\ldots,k\}$. 

\section{Asymptotic results}\label{sec:asymp}

In the previous section, we presented a procedure that allows to sample from
the posterior distribution of the parameter $\theta$ of a parametric model
$\{f_\theta,\theta\in\Theta\}$ given a sample of $N$ observations. In this 
section, we will discuss the asymptotic properties of the posterior 
distribution as the sample size $N$ tends to $\infty$.

The asymptotic analysis of Bayes procedures usually relies on the
Bernstein--von Mises theorem which allows for an asymptotic normal 
approximation of the posterior distribution of $\sqrt{N} (\theta-\theta_0)$, 
given the observations $z^{(1)},\ldots,z^{(N)}$ from the parametric model
$f_{\theta_0}$. The theorem then implies the asymptotic normality and 
efficiency of Bayesian point estimators such as the posterior mean or 
posterior median with the same asymptotic variance as the maximum likelihood
estimator.

\begin{theorem}[Bernstein-von Mises, Theorems 10.1 and 10.8 in \citet{vdV98}] \label{bernstein}
Let the parametric model $\{f_\theta,\theta\in\Theta\}$ be differentiable in
quadratic mean at $\theta_0$ with non-singular Fisher information matrix
$I_{\theta_0}$, and assume that the mapping $\theta \mapsto \sqrt{f_\theta(z)}$
is differentiable at $\theta_0$ for $f_{\theta_0}$-almost every $z$. Suppose 
that, for every $\varepsilon>0$, there exists a sequence of uniformly 
consistent tests $\phi_N=\phi_N(z^{(1)},\ldots,z^{(N)})\in\{0,1\}$ for
testing $H_0:\theta=\theta_0$ against $H_1:\|\theta-\theta_0\|_\infty \geq
\varepsilon$, that is
\begin{equation}\label{eq:uniftest}
\mathbb P_{\theta_0}(\phi_N = 1)=0 \quad \mbox{and} \quad \sup_{\|\theta-\theta_0\|_\infty \geq \varepsilon} \mathbb P_{\theta}(\phi_N = 0)=0 \quad \mbox{as } N\to\infty.
\end{equation}
where $\mathbb P_{\theta}$ denotes the probability measure induced by $N$ 
independent copies of $Z \sim f_\theta$. Suppose that the prior distribution 
$\pi_{\mathrm{prior}}(\mathrm{d}\theta)$ is absolutely continuous in a 
neighborhood of $\theta_0$ with a continuous positive density at $\theta_0$.
Then, under the distribution $f_{\theta_0}$, the posterior distribution 
satisfies
\[
\left\| \pi_{\mathrm{post}}(\mathrm{d}\theta\mid z^{(1)},\ldots,z^{(N)}) - \mathcal{N}\left(\theta_0+N^{-1/2}\Delta_{N,\theta_0},N^{-1}I_{\theta_0}^{-1} \right)\right\|_{TV}
  \stackrel{d}{\longrightarrow} 0 \quad\mbox{as } N\to\infty,
\]
where $\Delta_{N,\theta_0}={N}^{-1/2}\sum_{i=1}^N I_{\theta_0}^{-1} \partial_{\theta} \log f_{\theta_0} (z^{(i)})$ 
and $\| \cdot\|_{TV}$ is the total variation distance.

As a consequence, if the prior distribution $\pi_{\mathrm{prior}}(\mathrm{d}\theta)$
has a finite mean, the posterior median $\hat\theta_n^{Bayes}$ is asymptotically
normal and efficient, that is, it satisfies 
\begin{equation*}
\sqrt{N} (\hat\theta_N^{Bayes} - \theta_0) \stackrel{d}\longrightarrow \mathcal{N}(0,I_{\theta_0}^{-1}), \quad \mbox{as } N\to\infty.
\end{equation*}
\end{theorem}

In order to apply this theorem to max-stable distributions, two main 
assumptions are required: the differentiability in quadratic mean of the 
statistical model and the existence of uniformly consistent tests satisfying 
\eqref{eq:uniftest}. Differentiability in quadratic mean is a technical 
condition that imposes a certain regularity on the likelihood $f_{\theta_0}$.
For the case of multivariate max-stable models this property has been 
considered in detail in \cite{deo2017}, where equivalent conditions on the 
exponent function and the spectral density are given.

We now discuss the existence of uniformly consistent tests and propose a 
criterion based on pairwise extremal coefficients. This criterion turns out to
be simple and general enough since it applies for most of the standard models 
in extreme value theory. Indeed, in many cases, pairwise extremal coefficients
can be explicitly computed and allow for identifying the parameter~$\theta$.  

For a max-stable vector $Z$ with unit Fr\'echet margins, the pairwise extremal
coefficient $\tau_{i_1,i_2}\in [1,2]$ between margins $1\leq i_1< i_2\leq k$ is
defined by
$$ \mathbb{P}(Z_{i_1}\leq z,Z_{i_2} \leq z ) = \exp\left\{-\tau_{i_1,i_2}/z\right\}, \quad z>0. $$
It is the scale exponent of the unit Fr\'echet variable $Z_{i_1}\vee Z_{i_2}$
and hence satisfies
$$ \tau_{i_1,i_2}= \left(\mathbb{E}\left[\frac{1}{Z_{i_1}\vee Z_{i_2}}\right] \right)^{-1}.$$
In the case that $Z$ follows the distribution $f_\theta$, we write 
$\tau_{i_1,i_2}(\theta)$ for the associated pairwise extremal coefficient.

\begin{proposition}\label{prop:uniftest}
Let $\theta_0\in\Theta$ and $\varepsilon>0$. Assume that 
\begin{equation}\label{eq:uniftest1}
\inf_{\|\theta-\theta_0\|_\infty \geq \varepsilon} \max_{1\leq i_1<i_2\leq k} |\tau_{i_1,i_2}(\theta)-\tau_{i_1,i_2}(\theta_0)|>0.
\end{equation}
Then there exists a uniformly consistent sequence of tests $\phi_N$ satisfying
\eqref{eq:uniftest}.
\end{proposition}

\begin{remark}\label{rem:id}
  The identifiability of the model parameters $\theta\in\Theta$ through the 
  pairwise extremal coefficients $\tau_{i_1,i_2}(\theta)$, $1 \leq i_1 < i_2
  \leq k$ is a direct consequence of Equation~\eqref{eq:uniftest1}. 
\end{remark}

\begin{remark}\label{rem:mon}
  If $\theta = (\theta_1,\dots, \theta_p)\in\Theta$, and for any $1 \leq j 
  \leq p$ there exists $1 \leq i_1 < i_2 \leq k$, such that 
  $\tau_{i_1,i_2}(\theta)$ depends only on $\theta_j$ and it is strictly 
  monotone with respect to this component, then Equation~\eqref{eq:uniftest1}
  is satisfied.
\end{remark}

\begin{proof}[Proof of Proposition~\ref{prop:uniftest}]
For a random vector $Z$ with distribution $f_\theta$ and 
$1\leq i_1<i_2\leq k$, the random variable $1 / (Z_{i_1} \vee Z_{i_2})$
follows an exponential distribution with parameter $\tau_{i_1,i_2}(\theta)\in[1,2]$ 
and variance $\tau_{i_1,i_2}^{-2}(\theta) \in [1/4,1]$. Hence, 
\[
T_{i_1,i_2}^{-1} =\frac 1 N \sum_{i=1}^N \frac 1 {Z_{i_1}^{(i)} \vee Z_{i_2}^{(i)}}
\]
is an unbiased estimator of $\tau_{i_1,i_2}^{-1}(\theta)$ with variance less 
than or equal to $1/N$. Chebychev's inequality entails 
\begin{equation}\label{eq:uniftest2}
\mathbb{P}_\theta\left( \left| T_{i_1,i_2}^{-1}-\tau_{i_1,i_2}^{-1}(\theta)\right| > \delta \right)
\leq{} \frac 1 {N \delta^2}, \quad \mbox{for all } \delta > 0. 
\end{equation}
Define the test $\phi_N$ by
\[
\phi_N=\left\{\begin{array}{ll}0 & \mbox{if } \max_{1 \leq i_1 < i_2\leq k} \left|T_{i_1,i_2}^{-1} - \tau_{i_1,i_2}^{-1}(\theta_0) \right| \leq \delta, \\
                               1 &\mbox{otherwise}. \end{array}\right.
\]
We prove below that, for $\delta>0$ small enough, the sequence $\phi_N$ 
satisfies \eqref{eq:uniftest}. For $\theta=\theta_0$, the union bound together
with Eq.~\eqref{eq:uniftest2} yield 
\[
\mathbb{P}_{\theta_0}(\phi_B=1) \leq \sum_{1\leq i_1<i_2\leq k} \mathbb{P}_{\theta_0}\left( \left| T_{i_1,i_2}^{-1}-\tau_{i_1,i_2}^{-1}(\theta_0)\right| > \delta \right)
                                \leq \frac{k(k-1)}{2N\delta^2} \longrightarrow 0,
\]
as $N \to \infty$. On the other hand, Eq.~\eqref{eq:uniftest1} implies that 
there is some $\gamma>0$ such that
\[ \max_{1\leq i_1<i_2\leq k} \left|\tau_{i_1,i_2}^{-1}(\theta)-\tau_{i_1,i_2}^{-1}(\theta_0)\right|\geq \gamma,\quad \mbox{ for all } \|\theta-\theta_0\|_\infty \geq \varepsilon.
\]
Let $\theta\in\Theta$ be such that $\|\theta-\theta_0\|_\infty\geq\varepsilon$
and consider $1\leq i_1<i_2\leq k$ realizing the maximum in the above equation.
By the triangle inequality,
\[
\left|T_{i_1,i_2}^{-1}-\tau_{i_1,i_2}^{-1}(\theta)\right| 
\geq \left|\tau_{i_1,i_2}^{-1}(\theta)-\tau_{i_1,i_2}^{-1}(\theta_0)\right|-\left|T_{i_1,i_2}^{-1}-\tau_{i_1,i_2}^{-1}(\theta_0)\right|
\geq \gamma-\left|T_{i_1,i_2}^{-1}-\tau_{i_1,i_2}^{-1}(\theta_0)\right|,
\]
so that, on the event $\{\phi_N=0\}\subset \left\{ \left|T_{i_1,i_2}^{-1}-\tau_{i_1,i_2}^{-1}(\theta_0)\right| \leq \delta \right\}$,
we have
$$\left|T_{i_1,i_2}^{-1}-\tau_{i_1,i_2}^{-1}(\theta)\right|\geq \gamma-\delta.$$
Applying Eq.~\eqref{eq:uniftest2} again, we deduce, for $\delta\in (0,\gamma)$,
\[
\mathbb{P}_\theta(\phi_N=0)\leq \mathbb{P}_\theta\left(\left|T_{i_1,i_2}^{-1}-\tau_{i_1,i_2}^{-1}(\theta)\right| \geq \gamma - \delta\right) \leq \frac{1}{N(\gamma-\delta)^2}.
\]
Since the upper bound goes to $0$ uniformly in $\theta$ with $\|\theta-\theta_0\|_\infty \geq \varepsilon$,
as $N \to \infty$, this proves Eq.~\eqref{eq:uniftest}.
\end{proof}

\section{Examples} \label{sec:examples}

In the previous sections we discussed the practical implementation of a Markov 
chain Monte Carlo algorithm to obtain samples from the posterior distribution
$L\left(\theta,\{\tau^{(l)}\}_{l=1}^N \mid \{z^{(l)}\}_{l=1}^N\right)$ and 
its asymptotic behavior as $N \to \infty$. The only model-specific
quantity needed to run the algorithm are the weights $\omega(\tau_j,z)$ in~\eqref{weights}. In this section, we provide explicit formulas for these 
weights for several classes of popular max-stable models and prove that the
models satisfy the assumptions of the Bernstein--von Mises theorem; see Theorem~\ref{bernstein}. 
It follows that the posterior median $\hat\theta_N^{Bayes}$ is 
asymptotically normal and efficient for these models. 

For the calculation of the weights $\omega(\tau_j,z)$, we first note that all 
of the examples in this section admit densities as a simple consequence of 
Prop.~2.1 in \cite{deo2017}. We further note that, for the models considered in 
Subsections \ref{subsec:logistic}-\ref{subsec:BR}, we have $\lambda_I \equiv 0$
for all $I \subsetneq \{1,\ldots,k\}$, i.e.\ 
$$ \Lambda(A)= \int_{A} \lambda(z) \mu(\mathrm{d}z), \quad  A \subset [0,\infty)^k \setminus \{0\},$$
and, consequently, Equation \eqref{weights} simplifies to
$$  \omega(\tau_j, z) = \int_{(0,z_{\tau_j^c})} \lambda(z_{\tau_j},u_j) \mathrm{d}u_j. $$
For the posterior median $\hat\theta_n^{Bayes}$, in the sequel, we will always 
assume that the prior distribution is absolutely continuous with strictly 
positive density in a neighborhood of $\theta_0$, and that it has finite mean. 
Given the differentiability in quadratic mean of the model, it suffices to 
verify condition~\eqref{eq:uniftest1} in Prop.~\ref{prop:uniftest}. This 
implies the existence of a uniformly consistent sequence of tests and, by 
Remark~\ref{rem:id}, the identifiability of the model. Theorem \ref{bernstein} 
then ensures asymptotic normality of the posterior median. 

Analogously to the notation $z_I = (z_i)_{i \in I}$ for a vector $z \in 
\mathbb{R}^k$ and an index set $\emptyset \neq I \subset \{1,\ldots,k\}$, we
write $A_{I,J} = (A_{ij})_{i \in I, j \in J}$ for a matrix 
$A = (A_{ij})_{1 \leq i,j \leq k}$ and index sets $\emptyset \neq I, J \subset 
\{1,\ldots,k\}$. Proofs for the results presented in this section can be found
in Appendix \ref{app:proofs}.

\subsection{The logistic model} \label{subsec:logistic}

One of the simplest multivariate extreme value distributions is the logistic 
model where
\begin{align} \label{log_expo} 
  V(z)=\left(z_1^{-1/\theta}+\cdots+z_k^{-1/\theta}\right)^\theta, \quad \theta\in (0,1).
\end{align}
The logistic model is symmetric in its variables and interpolates between 
independence as $\theta \uparrow 1$ and complete dependence as $\theta 
\downarrow 0$.

\begin{proposition}\label{log_prop}
  Let $\tau = (\tau_1, \dots, \tau_\ell) \in \mathcal P_k$ 
  and $z\in E$. The weights $\omega(\tau_j,z)$ in \eqref{weights} for
  the logistic model with exponent measure \eqref{log_expo} are
  \begin{align}\label{log_weight}
  \omega(\tau_j,z)= \theta^{-|\tau_j|+1}\frac{\Gamma(|\tau_j|-\theta)}{\Gamma(1-\theta)}\,\left(\sum_{i=1}^k z_i^{-1/\theta}\right)^{\theta-|\tau_j|} \prod_{i\in\tau_j} z_i^{-1-1/\theta}.
  \end{align}
\end{proposition}

\begin{remark}
 From \ref{log_weight}, it can be seen that we can also write
\begin{eqnarray*}
  L(\tau,z) &=&\exp(-V(z))\left(\prod_{i=1}^k z_i^{-1-1/\theta}\right) \left(\sum_{i=1}^k z_i^{-1/\theta}\right)^{-k} \theta^{-k}\prod_{j=1}^\ell \tilde \omega(\tau_j,z)
\end{eqnarray*}
 with
\[
\tilde \omega(\tau_j,z)=\theta \frac{\Gamma(|\tau_j|-\theta)}{\Gamma(1-\theta)}\left(\sum_{i=1}^k z_i^{-1/\theta}\right)^{\theta}.
\]
This suggests to use the simplified weights $\tilde \omega$ for the Gibbs
sampler.
\end{remark}

\begin{proposition} \label{prop:cond-logistic}
 For the logistic model is differentiable with $\theta_0 \in (0,1)$, the 
 posterior median $\hat\theta_N^{Bayes}$ is asymptotically normal and 
 efficient as $N \to \infty$.
\end{proposition}

\subsection{The Dirichlet model} \label{subsec:dirichlet}

The Dirichlet model \citep{col1991} is defined by its spectral density $h$ 
on the simplex $S^{k-1}=\{w\in [0,\infty)^k\ :\ w_1+\cdots+w_k=1\}$. For 
parameters $\alpha_1,\ldots,\alpha_k>0$, it is given by
\begin{equation} \label{diri-spec}
 h(w) =  \frac 1 k \frac{\Gamma(1+\sum_{i=1}^k \alpha_i)}{(\sum_{i=1}^k \alpha_i w_i)^{k+1}} \prod_{i=1}^k \frac{\alpha_i}{\Gamma(\alpha_i)}  
           \left(\frac{\alpha_i w_i}{\sum_{j=1}^k \alpha_j w_j}\right)^{\alpha_i-1}, \quad w \in S^{k-1},
\end{equation}
and it has no mass on lower-dimensional faces of $S^{k-1}$ \citep{col1991}. Equivalently, the exponent function of the Dirichlet model is given by
\[
V(z)=k\mathbb{E}\left[ \max_{i=1,\ldots,k} \frac{W_i}{z_i}\right],
\]
where $W$ is a random vector with density $h(w)$.

\begin{proposition}\label{diri_prop}
  Let $\tau = (\tau_1, \dots, \tau_\ell) \in \mathcal P_k$ 
  and $z\in E$. The weights $\omega(\tau_j,z)$ in \eqref{weights} for
  the Dirichlet model with spectral density \eqref{diri-spec} are
  \begin{align}\label{diri_weight}
  \omega(\tau_j,z)= \prod_{i\in\tau_j} \frac{\alpha_i^{\alpha_i} z_i^{\alpha_i-1}}{\Gamma(\alpha_i)}  \int_0^\infty  e^{-\frac 1 r \sum_{i\in \tau_j} \alpha_i z_i}
  \left(\prod_{i\in\tau_j^c} F_{\alpha_i}(\alpha_i z_i/r)\right) r^{-2-\sum_{i=1}^k \alpha_i} \mathrm d r,
  \end{align}
  where 
  \begin{eqnarray*}
    F_{\alpha}(x) = \frac{1}{\Gamma(\alpha)} \int_0^x t^{\alpha -1 } e^{-t} \mathrm d t
  \end{eqnarray*}
  is the distribution function of a Gamma variable with shape $\alpha >0$.
\end{proposition}

\begin{proposition} \label{prop:cond-dirichlet}
 Consider the Dirichlet model with $\theta_0=(\alpha_1,\ldots,\alpha_k) \in
 \Theta = (0,\infty)^k$. 
 For $k \geq 3$ and almost every $\theta_0 \in \Theta$, the posterior median
 $\hat\theta_N^{Bayes}$ is asymptotically normal and efficient as 
 $N \to \infty$.
\end{proposition}

\begin{remark}
  We believe that the result for the posterior median holds true even for every
  $\theta_0 \in \Theta$. In the proof of \ref{prop:cond-dirichlet}, we need the
  partial derivatives of $(\alpha_1,\alpha_2)\mapsto \tau(\alpha_1,\alpha_2)$
  to be negative, but this can only be concluded almost everywhere.
\end{remark}

\subsection{The extremal-$t$ model and Schlather process} \label{subsec:extremalt}

The extremal-$t$ model \citep{nik2009, opi2013} is given by an exponent measure 
of the form
\begin{equation} \label{eq:exp-fctn-extr}
 V(z) = c_\nu \mathbb{E}\left[ \max_{i=1,\ldots,k} \frac{\max\{0,W_i\}^\nu}{z_i} \right],
\end{equation}
where $(W_1,\ldots,W_k)^\top$ is a standardized Gaussian vector with correlation
matrix $\Sigma$, $c_\nu=\sqrt{\pi} 2^{-(\nu-2)/2}\Gamma\{(\nu+1)/2\}^{-1}$ and
$\nu>0$. 

\begin{proposition}[\citep{thi2015b}] \label{extr_prop}
  Let $\tau = (\tau_1, \dots, \tau_\ell) \in \mathcal P_k$ and $z\in E$. The
  weights $\omega(\tau_j,z)$ in \eqref{weights} for the extremal-$t$ model with 
  exponent function \eqref{eq:exp-fctn-extr} are
  \begin{align}\label{extr_weight}
  \omega(\tau_j,z) ={}& T_{|\tau_j|+\nu}\left(z_{\tau_j^c}^{1/\nu} - \tilde \mu, \tilde \Sigma\right)
                    \cdot \nu^{1-|\tau_j|} \cdot \pi^{(1-|\tau_j|)/2} \cdot \det(\Sigma_{\tau_j,\tau_j})^{-1/2} \nonumber\\
        & \cdot \frac{\Gamma\{(\nu+|\tau_j|)/2\}}{\Gamma\{(\nu+1)/2\}}
          \cdot \prod\nolimits_{i \in \tau_j} |z_i|^{1/\nu-1} 
          \cdot \left\{ \left(z_{\tau_j}^{1/\nu}\right)^\top \Sigma^{-1}_{\tau_j,\tau_j} z_{\tau_j}^{1/\nu} \right\}^{-(\nu+|\tau_j|)/2}
  \end{align}
  where $\tilde\mu = \Sigma_{\tau_j^c, \tau_j} \Sigma^{-1}_{\tau_j} z_{\tau^j}^{1/\nu}$,
  $$\tilde \Sigma = (|\tau_j|+\nu)^{-1} \left(z_{\tau^j}^{1/\nu}\right)^\top \Sigma^{-1}_{\tau_j} z_{\tau_j} 
                 \left(\Sigma_{\tau_j^c} - \Sigma_{\tau_j^c,\tau_j} \Sigma^{-1}_{\tau_j,\tau_j} \Sigma_{\tau_j,\tau_j^c}\right)$$
  and $T_{k}(\cdot; \Sigma)$ denotes a multivariate Student distribution 
  function with $k$ degrees of freedom and scale matrix $\Sigma$.
\end{proposition}

\begin{proposition} \label{prop:extremalt}
 Consider the extremal-$t$ model with $\theta_0=(\Sigma,\nu)$ where $\Sigma$
 is a positive definite correlation matrix and $\nu>0$. Then, for fixed $\nu>0$
 the posterior median $\hat\theta_N^{Bayes}$ is asymptotically normal and 
 efficient as $N\to\infty$. 
\end{proposition}

\begin{remark}
  If $\nu$ is not fixed, then the parameter $\theta = (\Sigma, \nu)$ cannot be
  identified from the pairwise extremal coefficients and 
  Equation~\eqref{eq:uniftest1} is not satisfied. The identifiability can still
  be shown by considering the behavior of the bivariate angular measure at the 
  origin \citep[Section A.3.3]{eng2016b}.
\end{remark}

A popular model in spatial extremes is the extremal-$t$ process 
\citep{opi2013}, a max-stable process $\{Z(x), \ x \in \mathbb R^d\}$ whose 
finite-dimensional distributions $(Z(x_1),\ldots,Z(x_k))^\top$, $x_1,\ldots,x_k
\in \mathbb R^d$ have an exponent function of the form \eqref{eq:exp-fctn-extr}
where the Gaussian vector is replaced by a standardized stationary Gaussian 
process $\{W(x), \ x \in \mathbb R^d\}$ evaluated at $x_1,\ldots,x_k$. The
correlation matrix $\Sigma$ then has the form 
 $$ \Sigma = \{\rho(x_i-x_j)\}_{1 \leq i,j \leq k}, $$
where $\rho: \mathbb{R}^d \to [-1,1]$ is the correlation function of the 
Gaussian process $W$. The special case $\nu=1$ corresponds to the extremal 
Gaussian process \citep{sch2002}, also called Schlather process.

\begin{corollary}\label{cor:ext}
  Let $Z$ be a Schlather process on $\mathbb R^d$ with correlation function
  $\rho$ coming from the parametric family
  $$\rho(h) = \exp(-\|h\|_2^\alpha/s), \quad (s,\alpha) \in \Theta = (0,\infty)\times (0,2].$$
  Suppose that $Z$ is observed at pairwise distinct locations $t_1,\dots, t_k \in \mathbb R^d$ 
  such that not all pairs of locations have the same Euclidean distance. 
  Then, the posterior median of $\theta=(s,\alpha)$ is asymptotically normal.
\end{corollary}

\subsection{The H\"usler-Reiss model and the Brown--Resnick model} \label{subsec:BR}

The H\"usler--Reiss distribution \citep[cf.,][]{hue1989, kab2009} can be 
characterized by its exponent function
\begin{align} \label{eq:exp-fctn-BR}
  V(z) = \mathbb{E}\left[ \max_{i=1,\ldots,k} \frac{\exp\left\{W_i - \Sigma_{ii}/2\right\}}{z_i}\right],
\end{align}
where $W = (W_1,\ldots, W_d)^\top$ is a Gaussian vector with expectation $0$ and
covariance matrix $\Sigma$. It can be shown that the exponent function can be
parameterized by the matrix 
$$ \Lambda  = \{\lambda^2_{i,j}\}_{1\leq i,j\leq k} = \left\{ \frac 1 4 \mathbb{E}(W_i-W_j)^2 \right\}_{1 \leq i,j \leq k}$$
as we have the equality
\begin{align}\label{eq:exp-fctn-HR}
  V(z) = \sum_{p=1}^k  z_p^{-1} \Phi_{k-1}\left( 2\lambda^2_{p,-p}  + \log(z_{-p}/z_p); \Sigma^{(p)}\right), \quad z\in(0,\infty)^k,
\end{align}
\citep[cf.~][]{nik2009}, where for $p=1,\dots, k$, the matrix $\Sigma^{(p)}$ 
has $(i,j)$th entry $2(\lambda^2_{p,i} + \lambda^2_{p,j} - \lambda^2_{i,j})$, 
$i,j\neq p$ and $\Phi_{k-1}(\cdot, \Sigma^{(p)})$ denotes the $(k-1)$-dimensional
normal distribution function with covariance matrix $\Sigma^{(p)}$. 

Note that the positive definiteness of the matrices $\Sigma^{(p)}$, $p=1,\ldots,k,$ 
follows from the fact that $\Lambda$ is conditionally negative definite, i.e.\
\begin{align} \label{neg-def}
 \sum_{1\leq i,j\leq k} a_ia_j \lambda^2_{i,j} \leq 0
\end{align}
for all $a_1,\ldots,a_k \in \mathbb{R}$ summing up to $0$ \citep[cf.][Lem.~3.2.1]{ber1984}. 
In the following, we will assume that $\Lambda$ is even strictly positive definite,
i.e.\ equality in \eqref{neg-def} holds true if and only if $a_1 = \ldots = a_k=0$.
Then, all the matrices $\Sigma^{(p)}_{I,I}$ with $p \in \{ 1,\ldots,k\}$ and
$\emptyset \neq I \subset \{1,\ldots,k\}$ are strictly positive definite.

\begin{proposition}[\citep{wadsworth-tawn14}, \citep{asa2015} ] \label{br_prop}
  Let $\tau = (\tau_1, \dots, \tau_\ell) \in \mathcal P_k$ and $z\in E$. For $j\in \{1,\dots, k\}$, 
  choose any $p\in \tau_j$ and let $\tilde\tau = \tau_j \setminus \{p\}$, $\tilde\tau^c = \{1,\dots, k\} \setminus \tau_j$. 
  The weights $\omega(\tau_j,z)$ in \eqref{weights} for the H\"usler--Reiss
  distribution with exponent function \eqref{eq:exp-fctn-HR} are
  \begin{align}\label{BR_weight}
   \omega(\tau_j,z) = \frac{1}{z_p^2\prod_{i\in \tilde\tau} z_i}
  \varphi_{|\tilde \tau|}\left\{ z^*_{\tilde\tau}; \Sigma^{(p)}_{\tilde\tau,\tilde\tau}\right\}
  \Phi_{|\tilde\tau^c|}\left\{ z^*_{\tilde\tau^c} - \Sigma^{(p)}_{\tilde \tau^c,\tilde\tau} (\Sigma_{\tilde\tau,\tilde\tau}^{(p)})^{-1}z^*_{\tilde\tau}; \hat\Sigma^{(p)}\right\},
  \end{align}
  where 
  \begin{align*}
   z^* = \left\{\log\left(\frac{z_i}{z_p}\right) + \frac{\Gamma(x_i,x_p)} 2 \right\}_{i=1,\dots, k} \text{ and } 
    \hat \Sigma^{(p)} =  \Sigma^{(p)}_{\tilde\tau^c,\tilde\tau^c} - \Sigma^{(p)}_{\tilde \tau^c,\tilde\tau}(\Sigma^{(p)}_{\tilde\tau,\tilde\tau})^{-1}\Sigma^{(p)}_{\tilde\tau,\tilde\tau^c}.
  \end{align*}
  Here $\Phi_k(\cdot; \Sigma)$ denotes a  $k$-dimensional Gaussian distribution
  function with mean $0$ and covariance  matrix $ \Sigma$, and 
  $\varphi_k(\cdot; \Sigma)$ its density. The functions $\Phi_0$ and 
  $\varphi_0$ are set to be constant $1$.
\end{proposition}

\begin{proposition}\label{prop:cond-hr}
 For the H\"usler--Reiss model with $\theta_0 = \Lambda$ being a strictly
 conditionally negative definite matrix. Furthermore, the posterior median
 $\hat\theta_N^{Bayes}$ is asymptotically normal and efficient as $N \to \infty$.
\end{proposition}

H\"usler--Reiss distributions are the finite dimensional distributions of the
max-stable Brown--Resnick process, a popular class in spatial extreme value 
statistics. Here, the Gaussian vectors $(W_1,\ldots,W_k)^\top$ in 
\eqref{eq:exp-fctn-BR} are the finite-dimensional distributions of a centered
Gaussian process $\{W(x), \, x \in \mathbb{R}^d\}$ which is parameterized via a
conditionally negative definite variogram
$\gamma: \mathbb{R}^d \times \mathbb{R}^d \to [0,\infty), \ \gamma(x_1,x_2) = \mathbb{E}(W(x_1)-W(x_2))^2.$
If $W$ has stationary increments, we have that $\gamma(x_1,x_2) = \gamma(x_1-x_2,0) =: \gamma(x_1 - x_2)$ 
and the resulting Brown--Resnick is stationary \citep{bro1977,kab2009}. The 
most common parametric class of variograms belonging to Gaussian processes with
stationary increments is the class of fractional variograms, which we consider 
in the following corollary.

\begin{corollary} \label{cor:cond-br}
  Consider a Brown--Resnick process on $\mathbb{R}^d$ with variogram coming 
  from the parametric family 
  $$\gamma(h) = \|h\|_2^\alpha/s, \quad (s, \alpha) \in \Theta = (0,\infty)\times (0,2).$$
  Suppose that the process is observed on a finite set of locations $t_1,\dots, 
  t_m \in \mathbb{R}^d$ such that the pairwise Euclidean distances are not all
  equal. Then the posterior median of $\theta=(s,\alpha)$ is 
  asymptotically normal.
\end{corollary}

\section{Simulation Study} \label{sec:simu}

Let $z^{(l)} = (z_1^{(l)},\ldots, z_k^{(l)})$, $l=1,\ldots,N$, be $N$ 
realizations of a $k$-dimensional max-stable vector $Z$ whose distribution
belongs to some parametric family $\{F_\theta, \ \theta \in \Theta\}$. As
described in Section \ref{sec:metho}, including the partition $\tau^{(l)}$ 
associated to a realization $z^{(l)}$ in a Bayesian framework allows to obtain
samples from the posterior distribution $L(\theta \mid z^{(1)}, \dots, 
z^{(N)})$ of $\theta$ given the data. This procedure uses the full dependence 
information of the multivariate distribution $Z$. This is in contrast to 
frequentist maximum likelihood estimation for the max-stable vector $Z$, where
even in moderate dimensions the likelihoods are to complicated for practical 
applications. Instead, at the price of likelihood misspecification, it is 
common practice to use only pairwise likelihoods which are assumed to be 
mutually independent. The maximum pairwise likelihood estimator \citep{PRS10} 
is then 
\begin{align}\label{PL_est}
  \hat \theta_{\rm PL} = \argmax_{\theta \in \Theta} \sum_{l=1}^N \sum_{1 \leq i < j \leq k} \log f_{\theta; i,j}(z_i^{(l)},z_j^{(l)}),
\end{align}
where $f_{\theta; i,j}$ denotes the joint density of the $i$th and $j$th 
component of $Z$ under the model $F_\theta$. Using only bivariate information
on the dependence results in efficiency losses.

In this section, we analyze the performance of our proposed Bayesian estimator 
and compare it to $\hat \theta_{\rm PL}$ and other existing methods. Since the 
latter are all frequentist approaches, for a Markov chain whose stationary 
distribution is the posterior, we obtain a point estimator 
$\hat \theta_{\rm Bayes}$ of $\theta$ as the posterior median, i.e.,
\begin{align*}
  \hat\theta_{\rm Bayes} = \text{median}\left\{L(\theta|z^{(1)}, \dots, z^{(N)})\right\}.
\end{align*}
As the parametric model we choose the logistic distribution introduced in 
Subsection \ref{subsec:logistic} with parameter space $\Theta = (0,1)$ and 
uniform prior. This choice covers a range of situations from strong to very 
weak dependence. Other choices of parametric models will result in different 
efficiency gains but the general observations in the next sections should 
remain the same.

We note that other functionals of the posterior distribution can be used to 
obtain point estimators. Simulations based on the posterior mean, for instance,
gave very similar results, and we therefore restrict to the posterior median
in the sequel. Similarly, changing the prior distributions does not have a strong
effect on the posterior distribution for the sample sizes we consider; see also 
Section~\ref{subsec:posterior}.

\subsection{Max-stable data} \label{sec:max-stable}

We first take the marginal parameters to be fixed and known and quantify the
efficiency gains of $\hat \theta_{\rm Bayes}$ compared to 
$\hat \theta_{\rm PL}$. We simulate $N = 100$ samples $z^{(1)},\dots,z^{(N)}$ 
from the logistic distribution for different dimensions $k \in \{6,10,50\}$ 
and different dependence parameters $\theta = 0.1 \times i$, $i =1, \dots, 9$. For 
each combination of dimension $k$ and parameter $\theta$ we then run a Markov 
chain with length $1500$, where we discard the first $500$ 
steps as the burn-in time. The empirical median of the remaining $1000$ elements gives 
$\hat\theta_{\rm Bayes}$. The chain is sufficiently long to reliably estimate 
the posterior median; see also the mixing properties in Section \ref{subsec:posterior}.
The maximum pairwise likelihood estimator $\hat 
\theta_{\rm PL}$ is obtained according to \eqref{PL_est}. The whole procedure 
is repeated $1500$ times to compute the corresponding root mean squared errors shown
in Figure~\ref{rmse}.

As expected, the use of full dependence information substantially decreases the
root mean squared errors and thus increases the efficiency of the estimates. In 
extreme value statistics, where typically only small data sets are available, 
this allows to reduce uncertainty due to parameter estimation. The advantage of 
this additional information becomes stronger for both higher dimensions
and weaker dependence, analogously to the observations in \citet{HDG15}.
This behavior can to some extent be understood by the results in \cite{shi1995} 
on the Fisher information of the logistic distribution for different dimensions 
and dependence parameters. When $\theta\downarrow 0$, pairwise likelihood 
performs just as well as full likelihood, which is sensible since, up to a 
multiplicative constant, the pairwise likelihood equals the full likelihood 
when $\theta=0$.

It is interesting to note that the estimates $\hat\theta_{\rm Bayes}$ appear to 
be unbiased in almost all cases, whereas the pairwise estimator has a finite 
sample bias. 

\begin{figure}
  \centering
  \includegraphics[trim = 10mm 0mm 0mm 10mm,width= .33\textwidth]{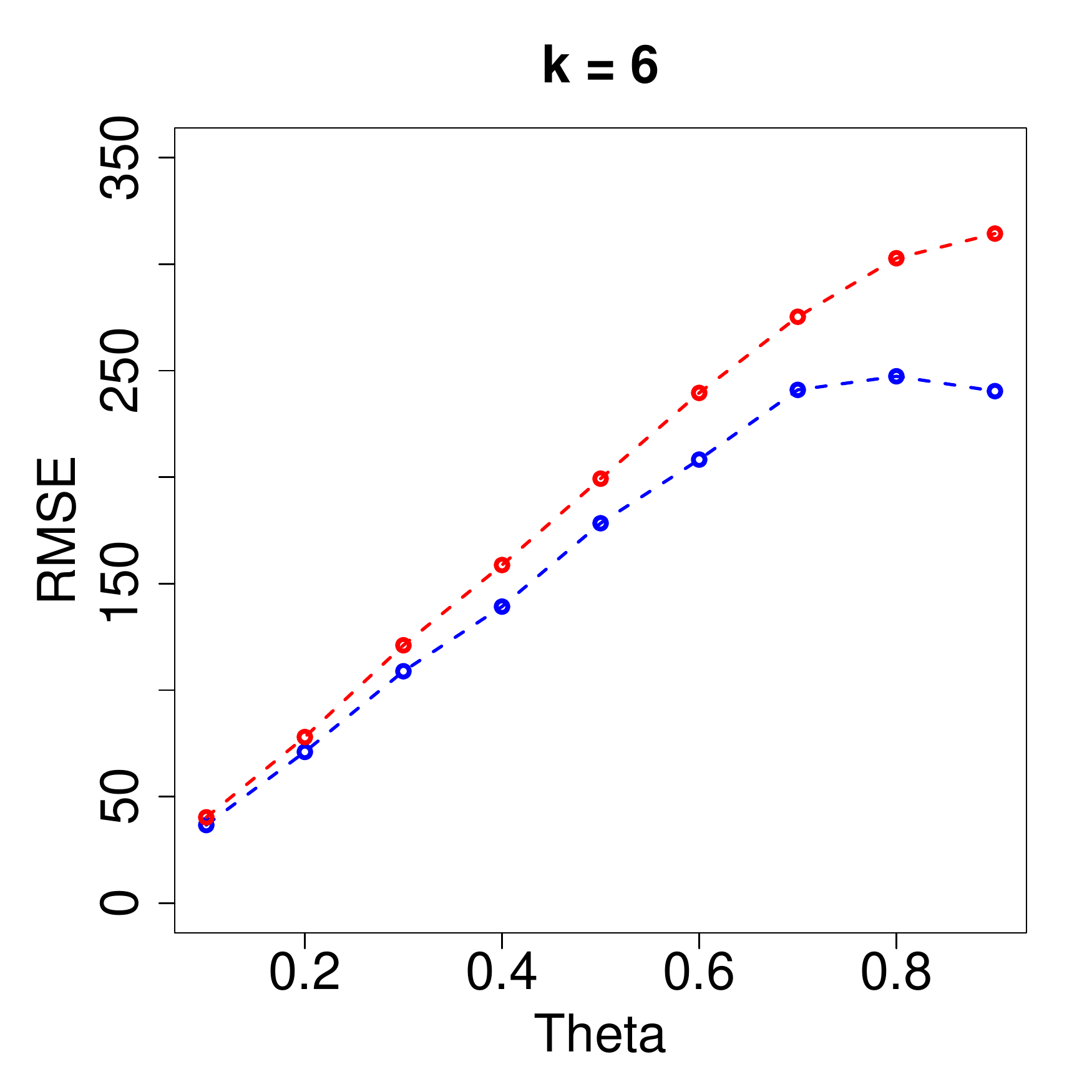}%
  \includegraphics[trim = 10mm 0mm 0mm 10mm,width= .33\textwidth]{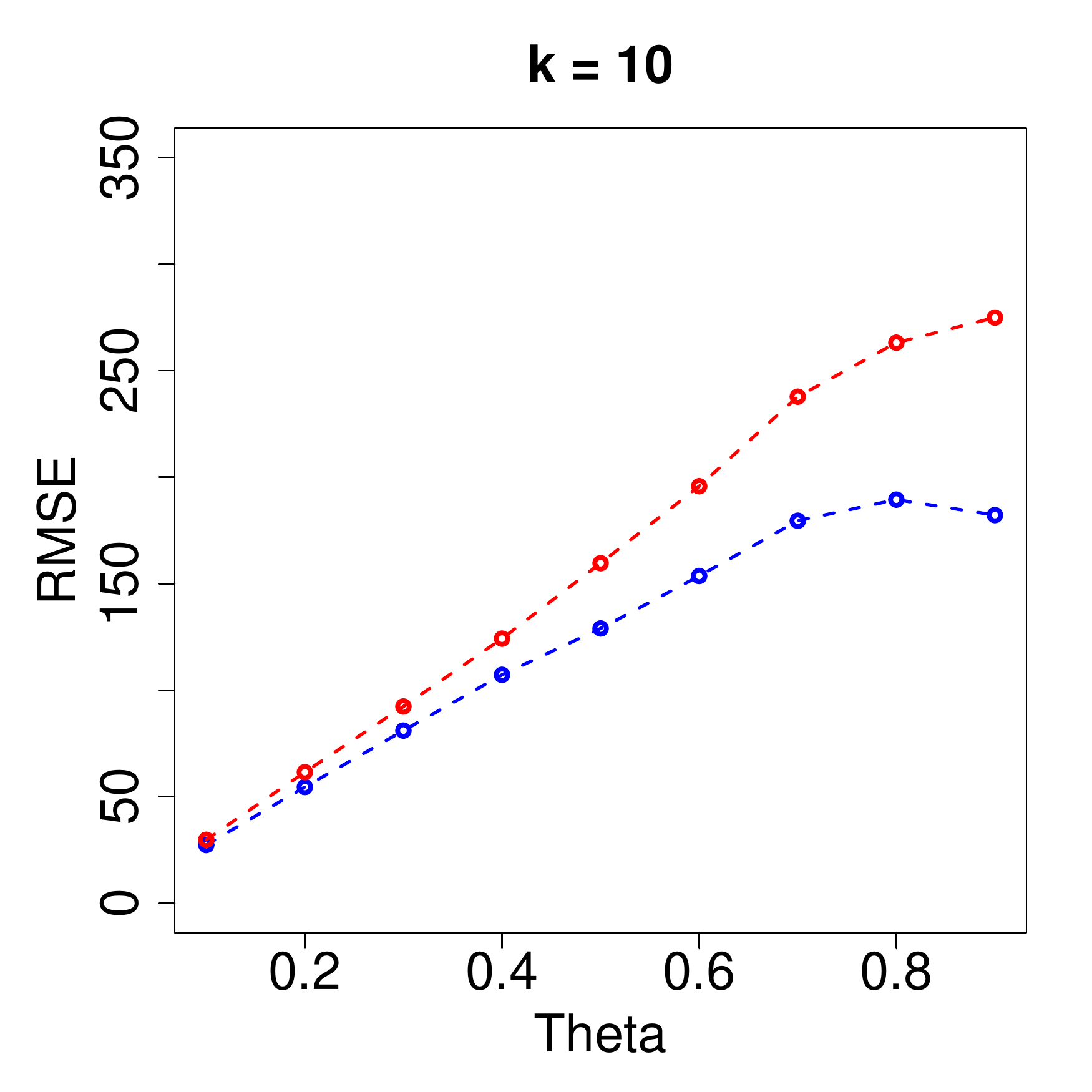}%
  \includegraphics[trim = 10mm 0mm 0mm 10mm,width= .33\textwidth]{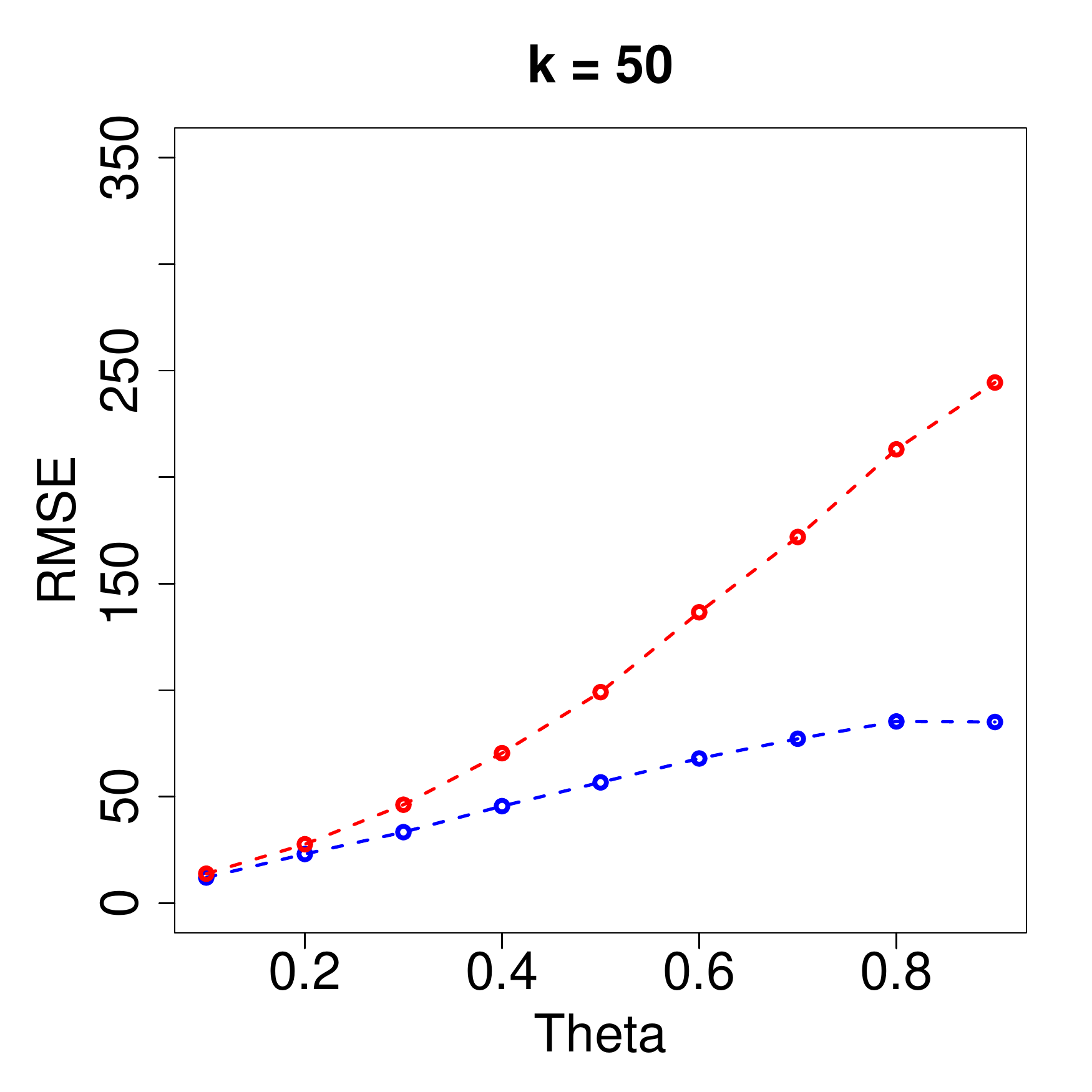}%
  \caption{Root mean squared errors (RMSE) of $\hat \theta_{\rm Bayes}$ (blue) and
           $\hat \theta_{\rm PL}$ (red) for different dimensions $k$ and different 
           parameters $\theta$. Values have been multiplied by $10000$.}
  \label{rmse}
\end{figure}

\subsection{Data in the max-domain of attraction}

In applications, the max-stable distribution $Z$ might not be observed exactly
but only as an approximation by componentwise block maxima of data vectors 
$X^{(1)},\dots, X^{(b)}$ in its max-domain of attraction with standard
Fr\'echet margins, where $b\in\mathbb N$ is the block size. Indeed, the random 
vector 
$$\tilde Z = \frac1b \left(\max_{l=1,\ldots,b} X_1^{(l)}, \ldots, \max_{l=1,\ldots,b} X_k^{(l)}\right),$$
approximates the distribution of $Z$, where the approximation improves for 
increasing $b$. In this situation we can associate to $\tilde Z$ the partition 
of occurrence times of the maxima, say $\tilde \tau$. \citet{ST05} proposed to 
use this information on the partition to simplify the likelihood of the 
max-stable distribution. For $N$ observations $\tilde z^{(1)}, \dots, 
\tilde z^{(N)}$ of $\tilde Z$ with partitions $\tilde \tau^{(1)}, \dots,\tilde \tau^{(N)}$
they defined the estimator
\begin{equation*}
  \hat \theta_{\rm ST} = \argmax_{\theta \in \Theta} \sum_{l=1}^N \log L(\tilde z^{(l)}, \tilde \tau^{(l)}; \theta).
\end{equation*}
This estimator suffers from two kinds of misspecification biases. Firstly, the 
$\tilde z^{(l)}$ are only approximately $Z$ distributed and, secondly, the 
partitions $\tilde \tau^{(l)}$ are only finite sample approximations to the true 
distribution of the limit partition $T$. For the latter, \citet{wadsworth14} 
proposed a bias  reduction method for moderate dimensions and showed in a 
simulation study that it significantly decreases the bias of the 
Stephenson--Tawn estimator in the case where the $X^{(k)}$ and thus also 
$\tilde Z$ follow exactly a max-stable logistic distribution. However, if the 
$X^{(k)}$ are samples from the outer power Clayton copula 
\citep[cf.][]{hof2011} and thus only in the max-domain of attraction of the 
logistic distribution, then even the bias reduced estimator suffers from 
significant bias \citep[cf.,][Table 3]{wadsworth14}.

We repeat the simulation study from Section \ref{sec:max-stable} with the only 
difference that, instead of sampling from $Z$, we simulate $N = 100$ samples 
$\tilde z^{(1)},\dots,\tilde z^{(N)}$ of $\tilde Z$, which is the rescaled 
maximum of $b=50$ samples from the outer power Clayton copula for different 
parameters. Based on these data in the max-domain of attraction of the logistic
distribution we estimate the dependence parameter $\theta$ using our Bayes 
estimator and compare it to the pairwise likelihood estimator. Both approaches 
ignore the additional information on the partitions $\tilde \tau^{(l)}$ that we 
have in this setup. On the other hand, we can also compute the Stephenson--Tawn
estimator and its bias reduced version by \citet{wadsworth14}, which explicitly 
include the partition information.

Table \ref{tab:simu-clayton} shows the root mean squared errors of the four 
estimators. All of them have a bias that plays a significant role for the 
overall estimation error and that is due to the model misspecification for only
approximately max-stable data. This bias is however much stronger for 
$\hat \theta_{\rm ST}$ and $\hat \theta_{\rm W}$, which use the again 
misspecified partitions. In this case, the Bayes estimator that treats the 
partitions as unknown and samples from them automatically seems to be more 
robust and does not need a bias correction. At the same time it has a small
variance and thus in many cases the smallest root mean squared error. Especially in 
higher dimensions ($\geq 20$) where the bias reduction of \citet{wadsworth14} 
can no longer be used, the Bayes estimator still provides a robust and 
efficient method of inference.

\begin{table} 
\begin{center}
\begin{tabular}{|l|rr|rr|rr|rr|} \hline
  & \multicolumn{2}{c|}{$\theta_0=0.1$} & \multicolumn{2}{c|}{$\theta_0=0.4$} & \multicolumn{2}{c|}{$\theta_0=0.7$}
  & \multicolumn{2}{c|}{$\theta_0=0.9$}\\
                                  $k$ & 6 & 10    & 6 & 10    & 6 & 10    & 6 & 10\\ \hline
$\rm{RMSE}(\theta_{\rm Bayes})$ & 36 & 29 &      144 & 111      & 241 & 191 &     262 & 220   \\ 
$ \rm{RMSE}(\theta_{\rm PL})$ & 40 & 32  &    159 & 127 &        279 & 235 &      311 & 286 \\ \hline
$ \rm{RMSE}(\theta_{\rm ST})$ & 38 & 29  &    148 & 126 &        352 & 401 &      647 & 840 \\
$ \rm{RMSE}(\theta_{\rm W})$ & 38 & 29  &    134 & 108 &        230 & 228 &      313 & 434 \\
\hline
\end{tabular}
\end{center}
\caption{Root mean squared errors of $\hat \theta_{\rm Bayes}$, $\hat \theta_{\rm PL}$, 
  $\hat \theta_{\rm ST}$ and $\hat \theta_{\rm W}$, estimated from 1500 estimates; 
  figures have been multiplied by $10000$.} \label{tab:simu-clayton}
\end{table}

\subsection{Estimation of marginal extreme value parameters}\label{est_margins}

In spatial settings, the marginal extreme value parameters are often 
estimated by using the independence likelihood \citep{CB07}, where all 
locations are assumed independent. This avoids to specify a dependence 
structure but can result in efficiency losses, even if only the marginal 
parameters are of interest.

We perform a simulation study to assess how using the full likelihoods in a 
Bayesian framework improves estimation of the marginal parameters. We fix the
dimension $k=10$ and set the marginal parameters to $\mu=1$, $\sigma=1$ and
$\xi\in\{-0.2,0.4,1\}$, equal for all $k$ margins. The dependence is logistic 
with unknown nuisance parameter $\theta_0\in\{0.1,0.4, 0.7, 0.9\}$.  

Based on $N=100$ independent samples from this model, we compare three
different estimation procedures. The first one is our Bayesian approach using
the full joint likelihood of the marginal parameters and the dependence 
parameter. We use a uniform prior for $\theta$, and independent normal priors
for $\mu$, $\log \sigma$ and $\xi$ with large standard deviations.
For the univariate case, more sophisticated choices for the prior distributions
are possible, including dependencies between the three extreme value parameters \citep[e.g.,][]{ST05a, nor2016}.

The second procedure is the maximum  pairwise likelihood estimator that only
uses bivariate dependence, and the third is the maximum independence likelihood 
estimator that completely ignores dependence between different components. Each
simulation and estimation is repeated 1500 times.

Table \ref{tab:margin_not_fix} contains the root mean squared errors of the marginal 
parameters for the three approaches. Interestingly, for the location and scale 
parameter we see only little difference between the three methods, meaning that 
they can be efficiently estimated without taking into account dependencies.
For the shape parameter, however, there are substantial improvements in the 
estimation error by including the unknown dependence structure in the model and
estimating it simultaneously. Since estimation of the shape is both the most 
difficult and the most important of the three extreme value parameters, the 
Bayesian approach is promising also for marginal tail estimation. Finally, we
observe that there is already an efficiency gain for the shape parameter when 
only the pairwise dependence is considered, but it is even more remarkable in the 
Bayesian setting with full likelihoods. Table \ref{tab:margin_not_fix} also shows
that these observations hold across different ranges
for the shape parameter $\xi$.

\begin{table} 
\begin{center}
\begin{tabular}{|l|rrr|rrr|rrr|rrr|} \hline
  & \multicolumn{3}{c|}{$\theta_0=0.1$} & \multicolumn{3}{c|}{$\theta_0=0.4$} & \multicolumn{3}{c|}{$\theta_0=0.7$} 
 & \multicolumn{3}{c|}{$\theta_0=0.9$}\\
   $\xi = -0.2$   & $\mu$ & $\sigma$ & $\xi$ & $\mu$ & $\sigma$ & $\xi$ & $\mu$ & $\sigma$ & $\xi$ & $\mu$ & $\sigma$ & $\xi$ \\ \hline
  Bayes full & 105 & 75 & 22 &      97 & 58 & 21 &     71 & 37 & 20 &     51 & 28 & 19\\ \hline
 Pairwise & 106 & 73 & 31 &         98 & 57 & 28 &      72 & 39 & 24 &    52 & 30 & 21 \\ \hline
 Independence  & 111 & 75 & 67 &   101 & 58 & 52 &     73 & 39 & 35 &     52 & 30 & 24  \\ \hline
\hline
   $\xi = 0.4$   & $\mu$ & $\sigma$ & $\xi$ & $\mu$ & $\sigma$ & $\xi$ & $\mu$ & $\sigma$ & $\xi$ & $\mu$ & $\sigma$ & $\xi$ \\ \hline
  Bayes full & 102 & 99 & 41 &      96 & 89 & 39 &     71 & 65 & 34 &     51 & 45 & 32\\ \hline
 Pairwise & 106 & 98 & 57 &         97 & 89 & 54 &      71 & 67 & 45 &     51 & 47 & 38 \\ \hline
 Independence  & 112 & 100 & 96 &   100 & 89 & 79 &     72 & 67 & 56 &     51 & 48 & 42  \\ \hline
 \hline
  $\xi = 1$  & $\mu$ & $\sigma$ & $\xi$ & $\mu$ & $\sigma$ & $\xi$ & $\mu$ & $\sigma$ & $\xi$ & $\mu$ & $\sigma$ & $\xi$ \\ \hline
  Bayes full & 109 & 155 & 85 &      100 & 144 & 76 &     77 & 111 & 59 &     53 & 75 & 46\\ \hline
 Pairwise & 106 & 146 & 94 &         96 & 135 & 90 &      74 & 108 & 72 &     52 & 75 & 53 \\ \hline
 Independence  & 110 & 146 & 127 &   98 & 135 & 107 &     74 & 109 & 81 &     52 & 76 & 57  \\ \hline
\end{tabular}
\end{center}
\caption{Root mean squared errors of $(\mu,\sigma,\xi)$ estimates with different values of $\xi$ for the Bayesian approach, 
 pairwise likelihoods and independence likelihoods, respectively, where $\theta$
 is an unknown nuisance parameter; figures have been multiplied by $1000$.} \label{tab:margin_not_fix}
\end{table}

\section{Applications in a Bayesian framework}\label{sec:bayes}

In the previous sections we discussed the efficiency gains of the Bayesian full
likelihood approach in the frequentist framework of point estimates. The Markov 
chain from Section \ref{subsec:bayes_setup} however produces not only a point
estimate but an estimate of the entire posterior distribution. For instance, 
this can directly be used to produce credible intervals for the parameter of 
interest. As a further application of our approach in the Bayesian framework,
we will present Bayesian model comparison in this section.

\subsection{The posterior distribution and credible intervals} \label{subsec:posterior}

As an illustration of the methodology we simulate a sample of $N = 15$ data
$z^{(l)} = (z_1^{(l)},\ldots, z_k^{(l)})$, $l=1,\ldots,N$, from a $k$-dimensional 
max-stable vector $Z$ whose distribution belongs to the parametric family of
logistic distributions introduced in Subsection \ref{subsec:logistic} with 
parameter space $\Theta = (0,1)$. We run the Markov chain from Subsection 
\ref{subsec:bayes_setup}. The left panel of Figure~\ref{MCs} shows the Markov
chain for the parameter $\theta$ with simulated data from the logistic 
distribution in dimension $k=10$ with $\theta_0 = 0.8$. The prior distribution 
is uniform, that is, $\pi_\theta = \text{Unif}(0,1)$. The chain seems to have 
converged to its stationary distribution, namely the posterior distribution 
\begin{equation}\label{posterior-theta}
  L\left(\theta \mid \{z^{(l)}\}_{l=1}^N\right) \propto \pi_\theta(\theta) \prod_{l=1}^N L(z^{(l)};\theta),
\end{equation}
after a burn-in period of about $200$ steps. The auto correlation of the Markov
chain in Figure \ref{acf} suggests that there is serial dependence up to a lag
of $30$ steps. The parallel chain that updates the partitions is difficult to
plot. The right panel of Figure \ref{MCs} therefore shows in each step as a 
summary the mean number $m$ of sets in the partitions $\tau^{(1)}, \dots, 
\tau^{(N)}$, that is, $m = 1/N \sum_{l=1}^N |\tau^{(l)}|$.
For complete independence ($\theta_0=1$) we must have $m = k = 10$, whereas for
complete dependence ($\theta_0=0$) we have $m = 1$.

\begin{figure}
  \centering
  \includegraphics[trim = 0mm 0mm 0mm 10mm,width= .45\textwidth]{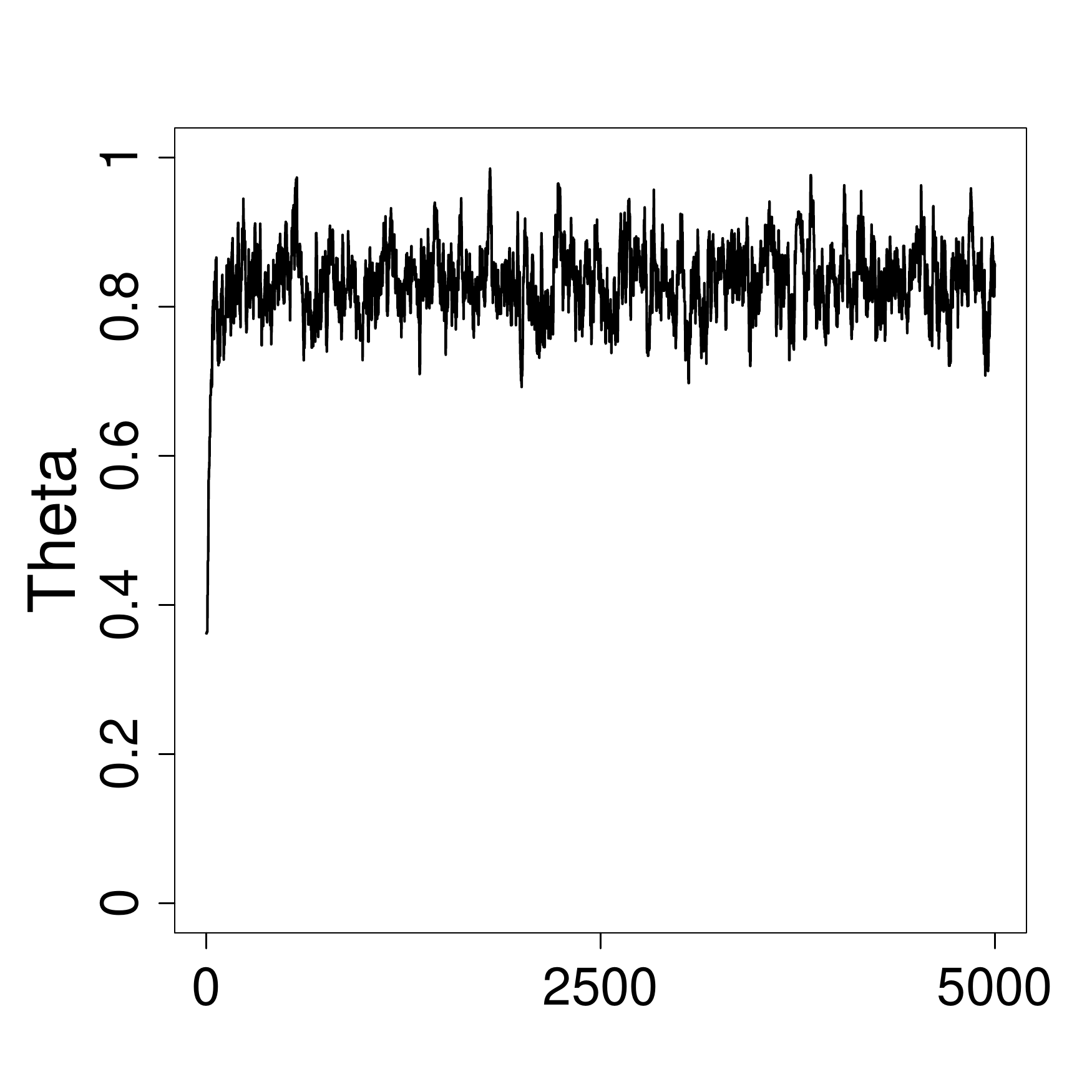}
  \includegraphics[trim = 0mm 0mm 0mm 10mm,width= .45\textwidth]{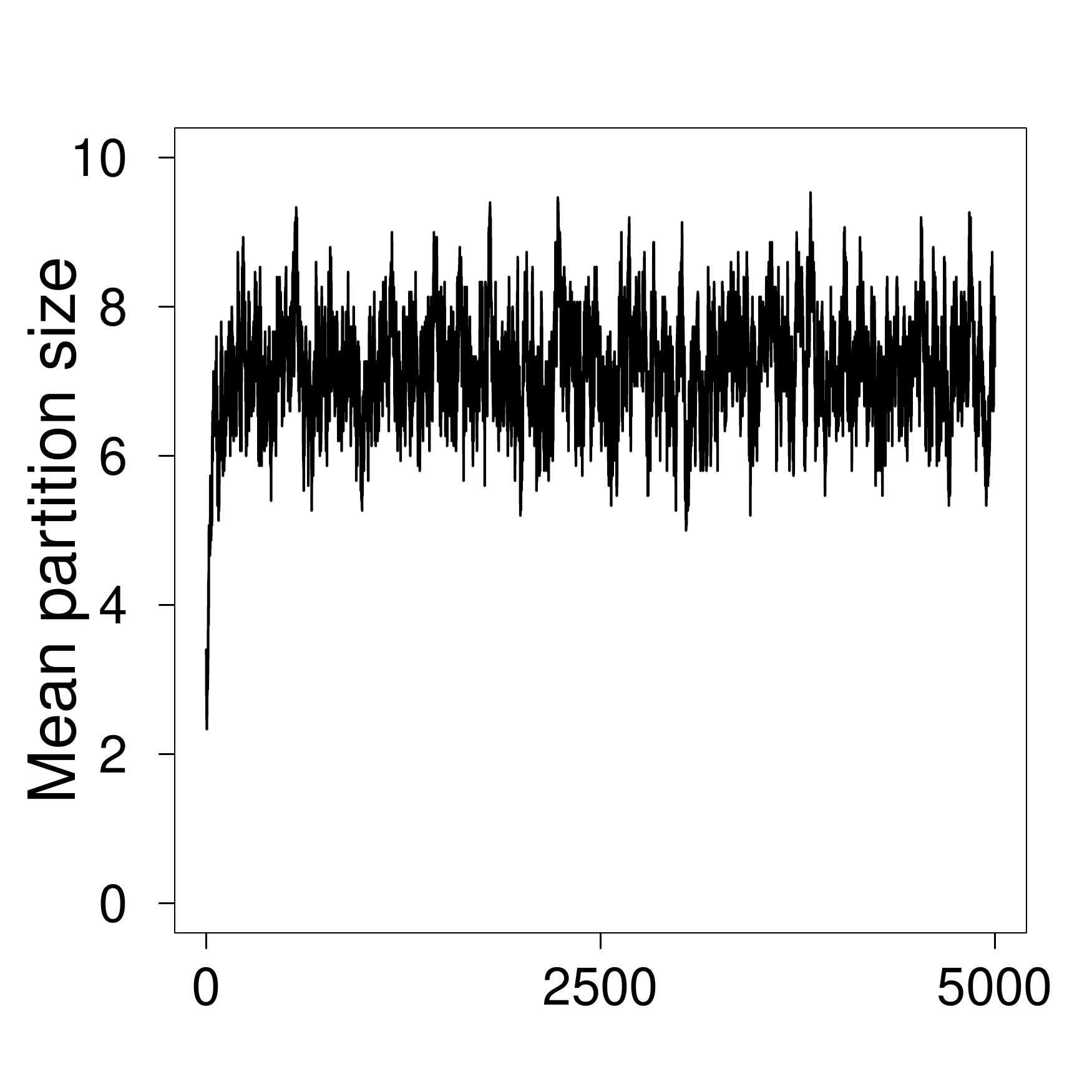}
  \caption{Markov chains for $\theta$ (left) and the mean partition size (right) with uniform prior.}
  \label{MCs}
\end{figure}
\begin{figure}[t]
  \centering
  \includegraphics[trim = 0mm 0mm 0mm 10mm,width= .45\textwidth]{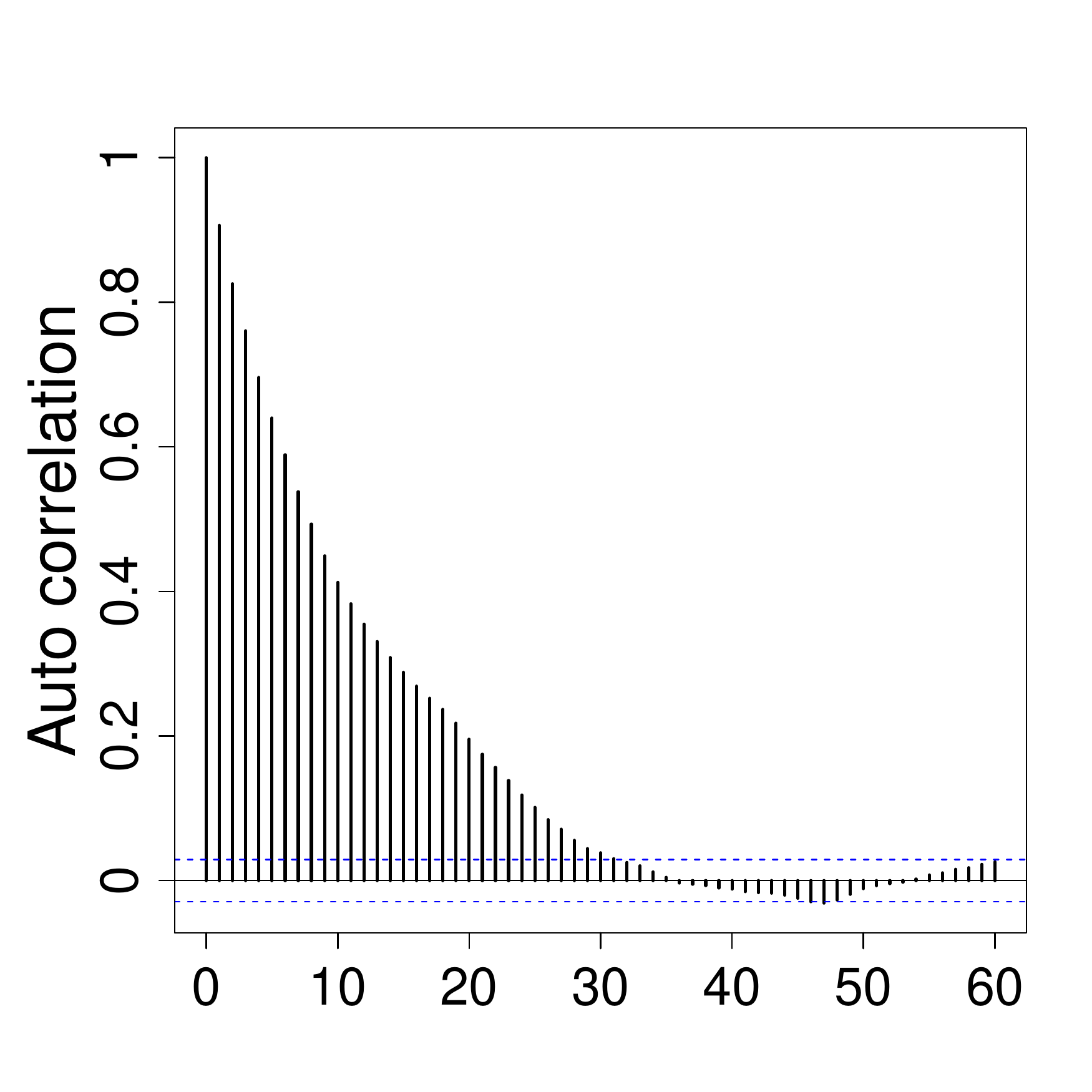}
  \caption{Auto correlation function of the Markov chain for the parameter $\theta$.}
  \label{acf}
\end{figure}

The left panel of Figure~\ref{hist} shows a histogram and an approximated 
smooth version of the posterior distribution, together with the uniform prior.
In order to assess the impact of the prior distribution on the posterior, the 
two other panels contain the corresponding plots for the same data set but for 
different priors, namely the beta distributions $\pi_\theta = 
\text{Beta}(0.5,0.5)$ (center) and $\pi_\theta = \text{Beta}(4,4)$ (right).
Even for a relatively small amount of $N=15$ data points, the influence of the
prior is not very strong.
\begin{figure}
 \centering
  \includegraphics[trim = 25mm 0mm 0mm 10mm,width= .3\textwidth]{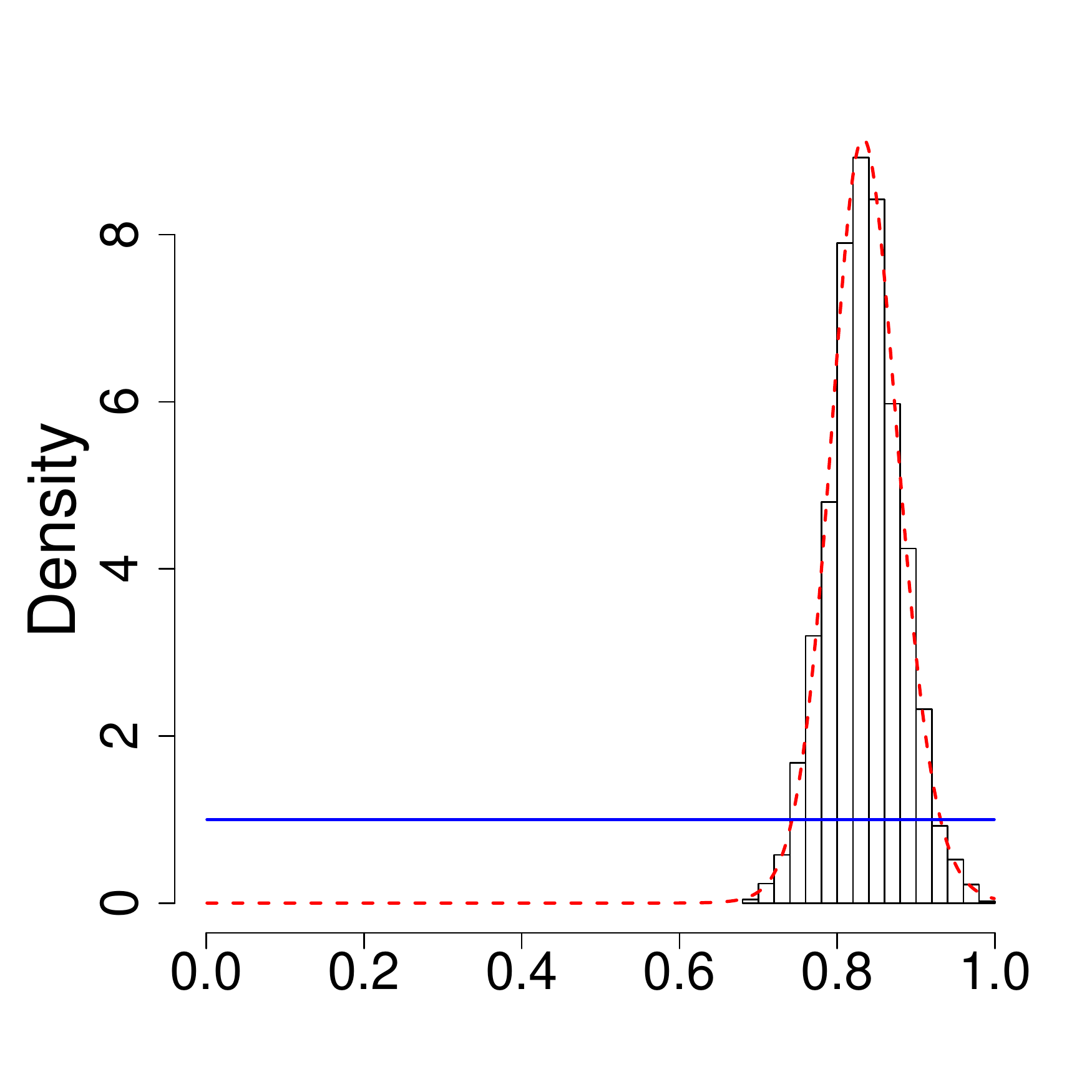}
  \includegraphics[trim = 25mm 0mm 0mm 10mm,width= .3\textwidth]{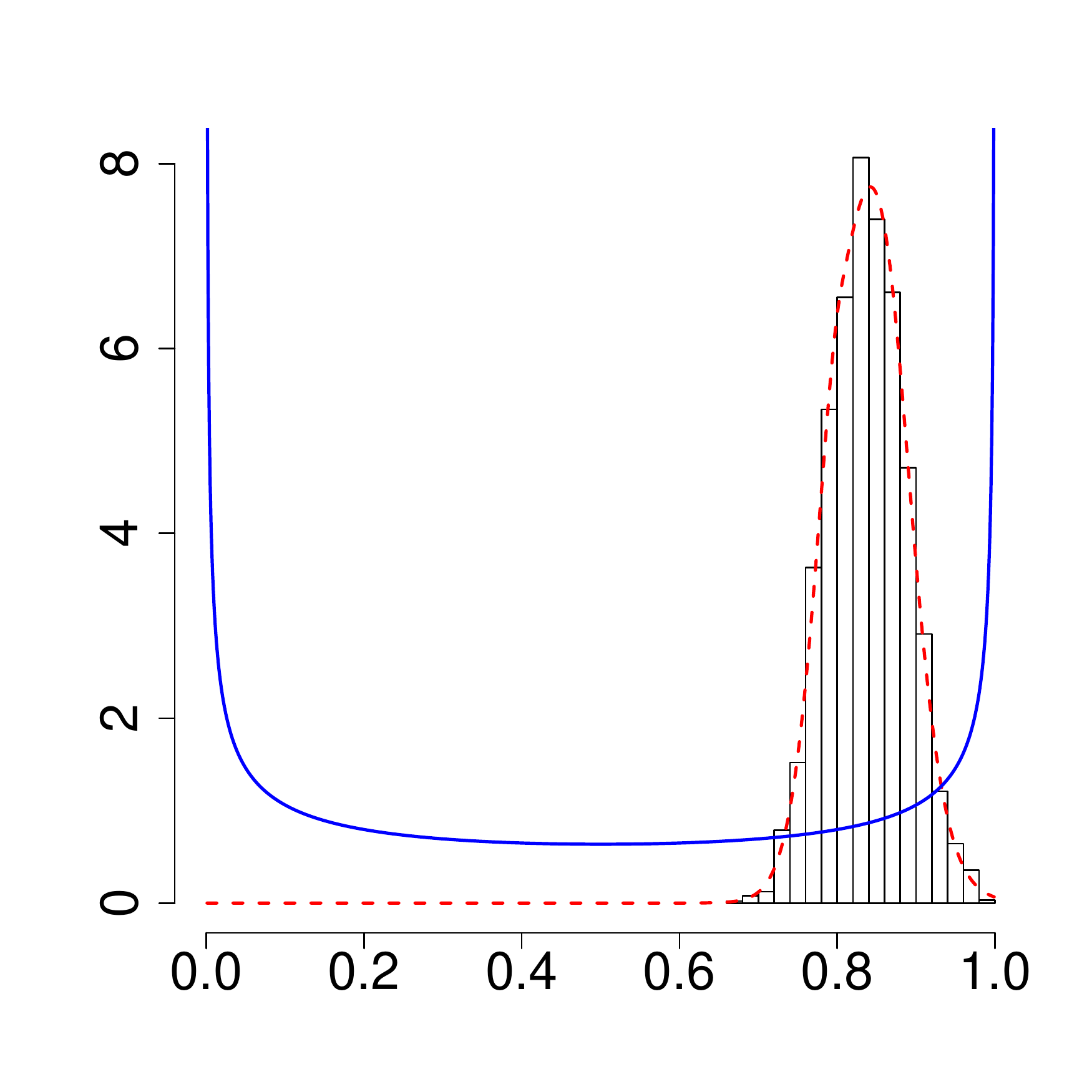}
  \includegraphics[trim = 25mm 0mm 0mm 10mm,width= .3\textwidth]{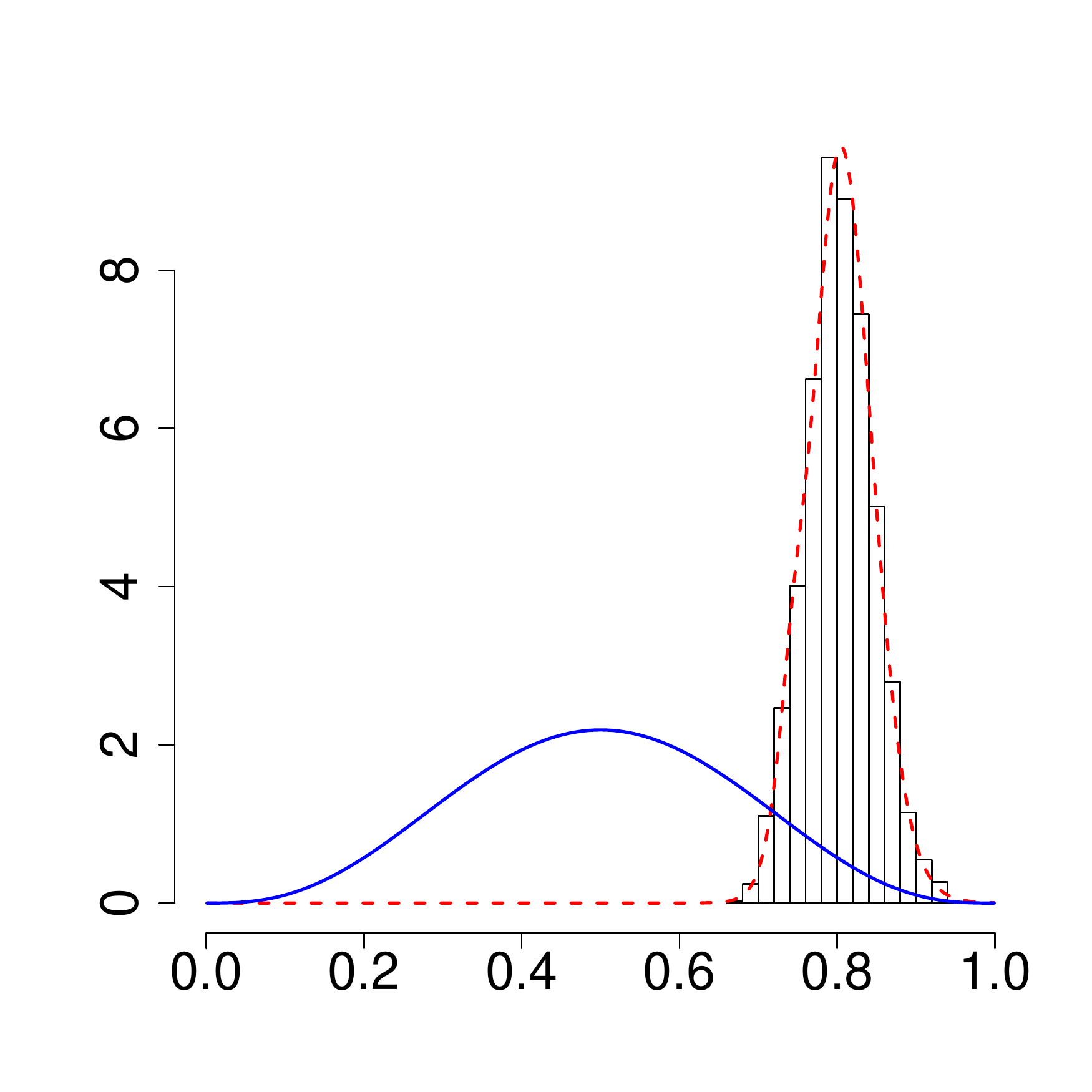}
  \caption{Histogram and smooth approximation of the posterior distribution (dotted red) 
           for different priors (solid blue): $\text{Unif}(0,1)$ (left), $\text{Beta}(0.5,0.5)$
           (center) and $\text{Beta}(4,4)$ (right).}
  \label{hist}
\end{figure}

The Bayesian setup provides us with a whole distribution for the parameter 
instead of a point estimate only. From this we can readily deduce credible
intervals for the parameter $\theta$. This is an advantage compared to 
frequentist composite likelihood methods since the Fisher information matrix
has a ``sandwich'' form adjusting for the misspecified likelihood, and 
confidence intervals are thus not easily computed \citep{PRS10}. When using 
composite likelihoods in a Bayesian setup, the posterior distributions are much
too concentrated and the empirical coverage rates are very small. Adjustments 
are necessary to obtain appropriate inference \citep{rib2012}. Since our 
approach uses the full, correct likelihood, no adjustment is needed to obtain 
accurate empirical coverage rates. Indeed, in Table \ref{credible} we provide
the coverage rates of the $95\%$ credible intervals obtained in the simulation
study in Section \ref{sec:max-stable} for some values of $\theta_0$. 
\begin{table}[b] 
\begin{center}
\begin{tabular}{|l|rrr|rrr|rrr|rrr|} \hline
  & \multicolumn{3}{c|}{$\theta_0=0.1$} & \multicolumn{3}{c|}{$\theta_0=0.4$} & \multicolumn{3}{c|}{$\theta_0=0.7$}
  & \multicolumn{3}{c|}{$\theta_0=0.9$}\\
  $k$ & 6 & 10 & 50   & 6 & 10 & 50    & 6 & 10 & 50   & 6 & 10 & 50\\ \hline
Coverage (in $\%$) & 94 & 93 & 90 &     95 & 94 & 94     & 94 & 94 & 94 &    94 & 94 & 90  \\ 
\hline
\end{tabular}
\end{center}
\caption{Empirical coverage rates of $95\%$ credible intervals obtained from
         the posterior distributions using full likelihood.} \label{credible}
\end{table}

\subsection{Bayesian model comparison}

Starting from data $z$ from a family of max-stable distributions
$\{F_\theta, \, \theta \in \Theta\}$, we consider two sub-models 
$M_1: \theta \in \Theta_1$ and $M_2: \theta \in \Theta_2$ for disjoint sets 
$\Theta_1, \Theta_2 \subset \Theta$. In Bayesian statistics, comparison of such
models is often based on the Bayes factor $B_{1,2}$, which translates the prior
odds into the posterior odds \citep[e.g.,][]{kas1995}, that is,
\begin{equation} \label{eq:bayes-factor}
 \frac{\pi_{{\rm posterior}}(\Theta_1)}{\pi_{{\rm posterior}}(\Theta_2)} = B_{1,2} \times \frac{\pi_{{\rm prior}}(\Theta_1)}{\pi_{{\rm prior}}(\Theta_2)}.
\end{equation}
The Bayes factor can also be written as $B_{1,2} = L (z \mid M_1) / L(z \mid M_2)$,
where
  \begin{equation} \label{eq:bayes-factor-int}
   L(z \mid M_i) = \int_\Theta L(z; \theta) \pi(\theta \mid M_i) {\rm d}\theta, \quad i=1,2,
  \end{equation}
are the so-called marginal probabilities of the data and $\pi(\cdot \mid M_i)$ 
is the prior density of the parameter $\theta$ under the model $M_i$. Since the
max-stable likelihood cannot be computed, the integral in 
\eqref{eq:bayes-factor-int} is computationally infeasible. However, we can use 
the estimation of the posterior probability \eqref{posterior-theta} discussed 
in the previous subsection and estimate
\begin{equation} \label{eq:bayes-factor-2}
B_{1,2}=\frac{\pi_\theta(\Theta_2)}{\pi_\theta(\Theta_1)} \times \frac{\int_{\Theta_1}L(\theta \mid \{z^{(l)}\}_{l=1}^N) {\rm d}\theta}{\int_{\Theta_2}L(\theta\mid \{z^{(l)}\}_{l=1}^N) {\rm d}\theta}.
\end{equation}

As an example, we consider a simple regression model
\begin{align}\label{regress}
  \xi_i = \alpha + i \beta, \quad i=1,\ldots,k,
\end{align}
for the marginal shape parameters $\xi_1,\ldots,\xi_k$ of the $k$-dimensional
max-stable distribution in dimension $k$. One might be interested in testing if
there is a linear trend in the shape parameters, and, thus, in comparing the 
models $M_1: \{\beta = 0\}$ and $M_2: \{\beta \neq 0\}$. In order to compute 
the Bayes factor as the ratio of the posterior probabilities of the two models
according to \eqref{eq:bayes-factor-2}, the prior distribution $\pi_\beta$ of
$\beta$ must be a mixture $0.5 \times \delta_{\{0\}} + 0.5 \times \pi_\beta^c$
of a Dirac point mass $\delta_{\{0\}}$ on $0$ and an appropriate continuous
distribution $\pi_\beta^c$ on $\mathbb{R}$. This ensures that we have a 
positive posterior probability on both sets $\{\beta = 0\}$ and 
$\{\beta \neq 0\}$ and the Bayes factor is well-defined.

Similarly as in Section \ref{est_margins}, we simulate $N=15$ data from a 
max-stable logistic distribution with dimension $k=10$ and dependence parameter
$\theta_0=0.5$, with marginal parameters $\mu_i=1$, $\sigma_i=1$ and $\xi_i$ as
in \eqref{regress} with $\alpha = 1$ and different values for $\beta$,
$i=1, \dots, k$. The prior distributions for the dependence, location and scale
parameters are chosen as in Section \ref{est_margins}. The prior for $\alpha$ 
is standard normal and for the prior for $\beta$ is a mixture of 
$0.5 \times \delta_{\{0\}} + 0.5 \times \pi_\beta^c$ of a point mass and a 
centered normal with standard deviation $0.5$ as the continuous component
$\pi_\beta^c$.

A Markov chain whose stationary distribution is the posterior distribution of
the parameters given the data can be constructed analogously to Section 
\ref{subsec:bayes_setup}. However, given the current state $\beta$ of the
Markov chain, the proposal $\beta^\ast$ is not drawn from a continuous 
distribution with density $q$, but from a mixture
$$p_0(\beta) \delta_{\{0\}}(\cdot) + (1-p_0(\beta)) q^c(\beta, \cdot)$$
of a Dirac point mass on $\{0\}$ and a continuous distribution with density
$q^c(\beta,\cdot)$, with mixture weight $p_0(\beta) \in (0,1)$. To ensure 
convergence of the Markov Chain, the densities $q^c(\beta,\cdot)$ should be 
chosen such that $q^c(\beta,\beta^\ast) > 0$ if and only if 
$q^c(\beta^\ast,\beta)>0$.

Figure \ref{BF} shows the Bayes factors $B_{1,2}$ that compare the model 
without trend $M_1: \{\beta = 0\}$ and the model with trend 
$M_2: \{\beta \neq 0\}$ for the simulated data described above. The true trend
varies from $\beta=0$, in which case $M_1$ would be correct, over positive
values up to $\beta=0.08$ where $M_2$ is the correct model. As comparison, we
implemented a Bayesian approach based on the independence likelihood 
\citep{CB07}, which is the product of the marginal densities and ignores the
dependence structure. The results show that using the full likelihood that 
takes the dependence into account and treats it as a nuisance parameter
significantly facilitates the distinction between the two different models. The
Bayes factors for the full likelihood show stronger support for $M_1$ if 
$\beta=0$, and decrease more rapidly to $0$ if $\beta>0$ than the Bayes factors
for independence likelihood. 

 \begin{figure}
\centering
  \includegraphics[trim = 25mm 0mm 0mm 10mm,width= .45\textwidth]{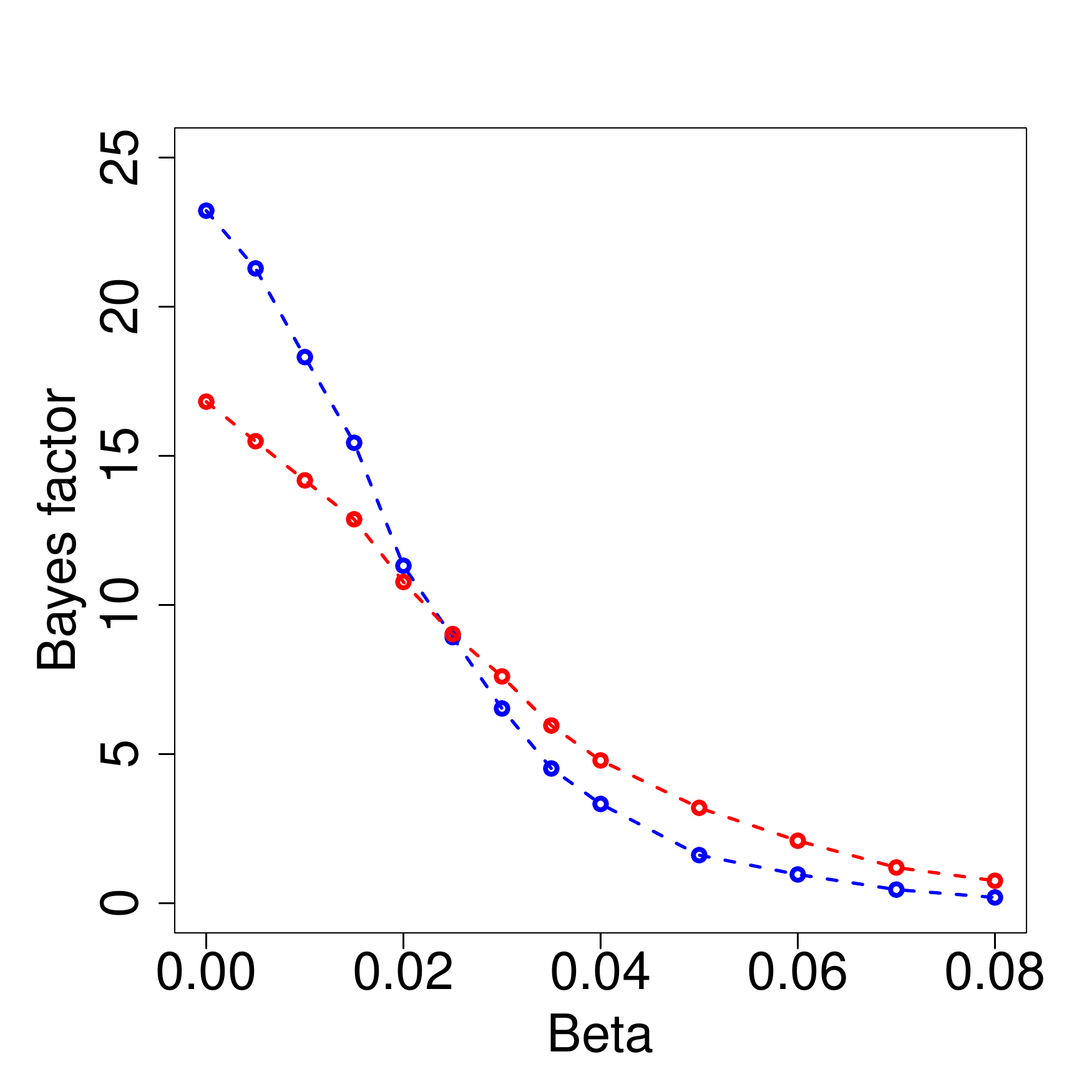}
  \caption{Bayes factors for different value of $\beta$ for full likelihood
           (blue) and independence likelihood (red).}
  \label{BF}
\end{figure}

Finally, we note that a similar approach has been proposed in the univariate
setting for estimation of the shape parameter in \citet{ST05a} in order to 
allow the Gumbel case $\xi = 0$ with positive probability.

\section{Discussion}\label{sec:disc}

We present an approach that allows for inference of max-stable distributions 
based on full likelihoods by perceiving the underlying random partition of the
data as latent variables in a Bayesian framework. The formulas for 
$\omega(\tau_j,z)$ provided in Section \ref{sec:examples} allow in principle 
to perform Bayesian inference based on full likelihoods for many popular 
max-stable distributions in any dimension. However, computational challenges 
arise for both the extremal-$t$ and the Brown--Resnick model in higher 
dimensions since the corresponding $\omega(\tau_j,z)$ require the evaluation of
a multivariate Student and Gaussian distribution functions, respectively, which 
have to be approximated numerically; see also \cite{thi2015}.

Making use of the weights $\omega(\tau_j,z)$, the posterior distribution of the
parameters becomes numerically available by samples based on Markov chain Monte 
Carlo techniques. As the results in Section \ref{subsec:posterior} indicate, 
the posterior distribution does not show strong influence of the prior 
distribution even in case of a rather small amount of data; cf., Figure 
\ref{hist}.  In most of the examples presented here, the proposal 
distributions for the model parameters in the Metropolis--Hastings algorithms
were chosen to be centered around the current state of the Markov chain with an
appropriate standard deviation, resulting in chains with satisfactory 
convergence and mixing properties; cf., Figures \ref{MCs} and \ref{acf}, for 
instance. Further improvements of these properties might be possible, e.g., by 
implementing an adaptive design of the Markov chain Monte Carlo algorithms.

In the frequentist framework, we propose to use the posterior median as a point
estimator for the model parameters. As the simulation studies in Section 
\ref{sec:simu} show, the use of full likelihoods considerably improves the 
estimation errors compared to the commonly used composite likelihood method
even in the case of a rather small sample size. This complements our theoretical
results on the asymptotic efficiency of the posterior median. Besides the point
estimator in the frequentist setting, we can also make use of the posterior
distribution in a Bayesian framework. In Section \ref{sec:bayes}, we discuss 
the use of credible intervals and Bayesian model comparison for max-stable 
distributions. Further applications such as Bayesian prediction 
are possible.

\appendix

\section{Proofs postponed from Section \ref{sec:examples}} \label{app:proofs}

\subsection{Proofs from Subsection \ref{subsec:logistic}}

\begin{proof}[Proof of Prop.\ \ref{log_prop}]
Taking partial derivative of the exponent function \eqref{log_expo} we obtain
\begin{eqnarray*}
   - \partial_{\tau_j} V_{\theta}(z)&=& \prod_{i=1}^{|\tau_j| -1} \left(\frac{i}{\theta} -1 \right)\left(\sum_{i=1}^k z_i^{-1/\theta}\right)^{\theta - |\tau_j|}
   \prod_{i\in\tau_j} z_i^{-1/\theta -1}. 
\end{eqnarray*}
We note that 
\[
\frac{\Gamma(|\tau_j|-\theta)}{\Gamma(1-\theta)}=\prod_{i=1}^{|\tau_j|-1}(i-\theta).
\]
Using this, Equation \eqref{eq:def_L_tau_frechet} becomes for the logistic model
\[
L(\tau,z;\theta)=\exp\{-V(z)\}\prod_{j=1}^\ell \omega(\tau_j,z) 
\]
with
\[
\omega(\tau_j,z)= \theta^{-1|\tau_j|+1}\frac{\Gamma(|\tau_j|-\theta)}{\Gamma(1-\theta)}\,\left(\sum_{i=1}^k z_i^{-1/\theta}\right)^{\theta-|\tau_j|} \prod_{i\in\tau_j} z_i^{-1-1/\theta}.
\]
\end{proof}

\begin{proof}[Proof of Prop.\ \ref{prop:cond-logistic}]
From Prop.~4.1 in \cite{deo2017} it follows that the model is differentiable  
in quadratic mean. For any $1 \leq i_1 < i_2 \leq k$, the pairwise extremal 
coefficient of the logistic model with parameter $\theta\in(0,1)$ is 
$\tau_{i_1,i_2}(\theta) = 2^\theta$, a strictly increasing function in 
$\theta$. The assertion of the proposition follows by Remark~\ref{rem:mon}.
\end{proof}

\subsection{Proofs from Subsection \ref{subsec:dirichlet}}

Both the proof of  Prop.~\ref{diri_prop} and Prop.~\ref{prop:cond-dirichlet}
rely on the following lemma.

\begin{lemma}\label{lem:dirichlet}
Let $Y(\alpha_1),\ldots,Y(\alpha_k)$ be independent random variables such that
$Y(\alpha)$ has a Gamma distribution with shape parameter $\alpha>0$ and scale
$1$.
\begin{itemize}
\item[(i)] Let $U_1>U_2>\dots$ be the points of a Poisson point process on 
$(0,\infty)$ with intensity $u^{-2}\mathrm{d}u$ and $\tilde Y^{(1)}, 
\tilde Y^{(2)}, \ldots$ independent copies of the random vector 
$\tilde Y = (Y(\alpha_i)/\alpha_i)_{1\leq i\leq k}$. Then the simple max-stable
random vector $Z=\bigvee_{i\geq 1} U_i\tilde Y^{(i)}$ has angular density 
\eqref{diri-spec}.
\item[(ii)] In the Dirichlet max-stable model \eqref{diri-spec}, the pair 
extremal coefficient $\tau_{i_1,i_2}$, $1\leq i_1<i_2\leq k$, is given by
\[
\tau_{i_1,i_2} = \tau(\alpha_{i_1},\alpha_{i_2}) = \mathbb{E}\left[ \frac{Y(\alpha_{i_1})}{\alpha_{i_1}}\vee\frac{Y(\alpha_{i_2})}{\alpha_{i_2}}\right].
\]
Furthermore, $\tau: (0,\infty)^2 \to [1,2]$ is continuously differentiable
and strictly decreasing in both components.
\end{itemize}
\end{lemma}

\begin{proof}
For the proof of the first part, we note that the intensity of the spectral
measure is given by
\[
\lambda(z)=\int_0^\infty f_{\tilde Y}(z/u)u^{-k-2}\mathrm{d}u,\quad z\in(0,\infty)^k,
\]
where 
\[
f_{\tilde Y}(\tilde y)=\prod_{i=1}^k \frac{\alpha_i^{\alpha_i}}{\Gamma(\alpha_i)}\tilde y_i^{\alpha_i-1}e^{-\alpha_i\tilde y_i},\quad \tilde y\in(0,\infty)^k,
\]
is the density of the random vector $\tilde Y$. A direct computation yields
\[
\lambda(z)= \frac{\Gamma(1+ \sum_{i=1}^d \alpha_i)}{(\sum_{i=1}^d \alpha_i z_i )^{1+\sum_{i=1}^d \alpha_i }} \prod_{i=1}^d  \frac{\alpha_i^{\alpha_i} z_i^{\alpha_i-1}}{\Gamma(\alpha_i)} .
\]
We see that the restriction of $\lambda$ to the simplex  $S^{k-1}$ is equal to 
$h$ which proves the claim.

The first statement of the second part, is a direct consequence of the first
part since
\[
\tau(\alpha_1, \alpha_2)=-\log\mathbb{P}(Z_1 \leq 1,Z_2\leq 1)=\mathbb{E}\left[ \frac{Y(\alpha_1)}{\alpha_1}\vee\frac{Y(\alpha_2)}{\alpha_2}\right].
\]
The proof of the strict monotonicity relies on the notion of convex order, see
Chapter 3.4 in \citet{DDGK05}. For two real-valued random variables $X_1$, 
$X_2$ we say that $X_1$ is lower than $X_2$ in convex order if
$\mathbb{E}[\varphi(X_1)]\leq  \mathbb{E}[\varphi(X_1)]$ for all convex 
functions $\varphi:\mathbb{R}\to\mathbb{R}$ such that the expectations exist. 
It is known that the family of random variables $(Y(\alpha)/\alpha)_{\alpha>0}$ 
is non-increasing in convex order \citep[Section 4.3]{ROS00}, and, in this 
case, the Lorenz order is equivalent to the convex order 
\citep[Property 3.4.41]{DDGK05}. We show below that this implies that 
$\tau(\alpha_1,\alpha_2)$ is strictly decreasing in its arguments.

Let $\alpha_1'> \alpha_1>0$ and $\alpha_2>0$ and let us prove that 
$\tau(\alpha_1',\alpha_2)<\tau(\alpha_1,\alpha_2)$. For independent random 
variables $Y(\alpha_1)$, $Y(\alpha_1')$ and $Y(\alpha_2)$, we have 
\[
\tau(\alpha_1',\alpha_2)=\mathbb{E}\left[ \frac{Y(\alpha_{1}')}{\alpha_1'}\vee\frac{Y(\alpha_2)}{\alpha_{2}}\right]  
\quad \mbox{and}\quad \tau(\alpha_1,\alpha_2)=\mathbb{E}\left[ \frac{Y(\alpha_{1})}{\alpha_{1}}\vee\frac{Y(\alpha_{2})}{\alpha_{2}}\right].
\]   
Using that $Y(\alpha_1')/\alpha_1'$ is lower than $Y(\alpha_1)/\alpha_1$ in 
convex order, we obtain
\begin{equation}\label{eq:cvx-order}
\mathbb{E}\left[ \frac{Y(\alpha_{1}')}{\alpha_{1}'}\vee\frac{y_{2}}{\alpha_{2}}\right]\leq \mathbb{E}\left[ \frac{Y(\alpha_{1})}{\alpha_{1}}\vee\frac{y_{2}}{\alpha_{2}}\right]\quad \mbox{for all } y_2>0,
\end{equation}
because the map $u\mapsto u\vee(y_2/\alpha_2)$ is convex. Replacing $y_2$ by
$Y(\alpha_2)$ and integrating, we get 
$\tau(\alpha_1',\alpha_2)\leq \tau(\alpha_1,\alpha_2)$. The equality 
$\tau(\alpha_1',\alpha_2)= \tau(\alpha_1,\alpha_2)$ would imply that for
almost every $y_2>0$ the equality holds in \eqref{eq:cvx-order} which is true 
if and only if $Y(\alpha_1)/\alpha_1$ and $Y(\alpha_1')/\alpha_1'$ have the
same distribution. Since this is not the case, 
$\tau(\alpha_1',\alpha_2)< \tau(\alpha_1,\alpha_2)$ and $\tau$ is strictly
decreasing in $\alpha_1$. By symmetry, $\tau$ is also strictly decreasing in 
$\alpha_2$.

Finally, the fact that $(\alpha_1,\alpha_2)\mapsto \tau(\alpha_1,\alpha_2)$ is
continuously differentiable follows from the integral representation
\begin{equation}\label{eq:tau_alpha}
\tau(\alpha_1,\alpha_2)=\int_0^\infty\int_0^\infty \frac{y_1}{\alpha_1}\vee \frac{y_2}{\alpha_2}\frac{1}{\Gamma(\alpha_1)\Gamma(\alpha_2)}y_1^{\alpha_1-1}y_2^{\alpha_2-1}e^{-y_1}e^{-y_2}\,\mathrm{d}y_1\mathrm{d}y_2,
\end{equation}
for $\alpha_1,\alpha_2>0$, and standard theorems for integrals depending on a
parameter.
\end{proof}

\begin{proof}[Proof of Prop.~\ref{diri_prop}]
From the construction given in the first part of Lemma \ref{lem:dirichlet},
we obtain
\begin{eqnarray*}
  \lambda(z)
&=&\int_{0}^\infty   \left( \prod_{i=1}^k \frac{\alpha_i^{\alpha_i}}{\Gamma(\alpha_i)}\left(\frac{z_i}{r}\right)^{\alpha_i-1}e^{-(\alpha_i z_i/r)} \right) r^{-2-k}\mathrm{d}r\\
&=&\prod_{i=1}^k \frac{\alpha_i^{\alpha_i}z_i^{\alpha_i-1}}{\Gamma(\alpha_i)} \int_{0}^\infty e^{-\frac1r \sum_{i=1}^k \alpha_i z_i}  r^{-2-\sum_{i=1}^k \alpha_i}\mathrm{d}r
\end{eqnarray*}
and, consequently,
\begin{eqnarray*}
&&\int_{u_j<z_{\tau_j^c}}\lambda(z_{\tau_j},u_j)\mathrm{d}u_j\\
&=& \prod_{i \in \tau_j}  \frac{\alpha_i^{\alpha_i}z_i^{\alpha_i-1}}{\Gamma(\alpha_i)} 
    \int_{0}^\infty  \int_{u_j<z_{\tau_j^c}}\ \prod_{i \in \tau_j^c} \left(\frac{\alpha_i^{\alpha_i} z_i^{\alpha_i-1}}{\Gamma(\alpha_i)}
      e^{-(\alpha_i z_i/r)} \right) \mathrm d u_j \cdot e^{-\frac 1 r \sum_{i \in \tau_j} \alpha_i z_i} r^{-2- \sum_{i=1}^k \alpha_i}\mathrm{d}r\\
&=& \prod_{i\in\tau_j} \frac{\alpha_i^{\alpha_i} z_i^{\alpha_i-1}}{\Gamma(\alpha_i)}  \int_0^\infty  e^{-\frac 1 r \sum_{i\in \tau_j}\alpha_i z_i}
  \left(\prod_{i\in\tau_j^c} F_{\alpha_i}(\alpha_i z_i/r)\right) r^{-2-\sum_{i \in \tau_j} \alpha_i} \mathrm d r,
\end{eqnarray*}
where
\begin{eqnarray*}
  F_{\alpha}(x) = \frac{1}{\Gamma(\alpha)} \int_0^x t^{\alpha -1 } e^{-t} \mathrm d t,
\end{eqnarray*}
is the distribution function of a Gamma variable with shape $\alpha >0$.
\end{proof}

\begin{proof}[Proof of Prop.~\ref{prop:cond-dirichlet}]
 Prop.~4.2 in \cite{deo2017} implies that the model is differentiable in
 quadratic mean. In order to verify Eq.~\eqref{eq:uniftest1} for the Dirichlet
 model, we consider the mapping
 $$ \Psi: \ (0,\infty)^k \to [1,2]^k, \ \theta = (\alpha_1,\ldots,\alpha_k) \mapsto (\tau_{1,2},\tau_{2,3}, \tau_{1,3}, \tau_{1,4}, \ldots, \tau_{1,k}).$$
 We first show that $\Psi$ is injective. To this end, let $\theta^{(1)} \neq \theta^{(2)} \in \Theta$
 where $\psi_i = (\alpha_1^{(i)},\ldots,\alpha_k^{(i)})$, $i=1,2$. We distinguish
 between two cases. First, we assume that $\theta^{(1)}$ and $\theta^{(2)}$ share
 at least one common component. Then, there is a pair $(i,j) \in \{(1,2), (2,3), (1,3), (1,4), \ldots, (1,k)\}$
 such that $(\alpha^{(1)}_i, \alpha^{(1)}_j)$ and $(\alpha^{(2)}_i, \alpha^{(2)}_j)$
 differ in exactly one component. As $\tau_{i,j} = \tau(\alpha_i,\alpha_j)$ is 
 strictly decreasing both in $\alpha_i$ and $\alpha_j$, by Lemma 
 \ref{lem:dirichlet}, we have that 
 $\tau(\alpha^{(1)}_i,\alpha^{(1)}_j) \neq \tau(\alpha^{(2)}_i,\alpha^{(2)}_j)$.
 Secondly, we consider the case that $\theta^{(1)}$ and $\theta^{(2)}$
 do not share any common component. Then, there is a pair
 $(i,j) \in \{(1,2), (2,3), (1,3)\}$ such that both components
 $(\alpha^{(1)}_i - \alpha^{(2)}_i, \alpha^{(1)}_j - \alpha^{(2)}_j)$
 have the same sign and, again, by the strict monotonicity of
 $\tau(\alpha_i,\alpha_j)$, it follows that 
 $\tau(\alpha^{(1)}_i,\alpha^{(1)}_j) \neq \tau(\alpha^{(2)}_i,\alpha^{(2)}_j)$.
 Hence, in both cases, $\Psi(\theta^{(1)}) \neq \Psi(\theta^{(2)})$, that is, 
 $\Psi$ is injective and there exists a unique inverse function 
 $\Psi^{-1}: \Psi((0,\infty)^k) \to (0,\infty)^k$.
 
Consider the set
 $$ \Theta' = \{(\alpha_1,\ldots,\alpha_k) \in (0,\infty)^k: \ \partial_{\alpha_i} \tau(\alpha_i,\alpha_j) < 0 
                                                  \text{ and } \partial_{\alpha_j} \tau(\alpha_i,\alpha_j) < 0
                \text{ for all } 1 \leq i < j \leq k\}.$$
Note that since $\tau$ is continuously differentiable and strictly decreasing 
in its argument, $\Theta\setminus\Theta'$ has Lebesgue measure $0$. For all 
$\theta_0 =(\alpha_1,\ldots,\alpha_k) \in \Theta'$, the Jacobian $D \Psi$ 
satisfies
 \begin{align*}
\det\{D\Psi(\theta_0)\} =& \left\{\partial_{\alpha_1} \tau(\alpha_1,\alpha_2) \cdot \partial_{\alpha_2} \tau(\alpha_2,\alpha_3) \cdot \partial_{\alpha_3} \tau(\alpha_1,\alpha_3)\right.\\
 & \hfill \left.                      + \partial_{\alpha_2} \tau(\alpha_1,\alpha_2) \cdot \partial_{\alpha_3} \tau(\alpha_2,\alpha_3) \cdot \partial_{\alpha_1} \tau(\alpha_1,\alpha_3) \right\}
                 \cdot \prod_{j=4}^k \partial_{\alpha_j} \tau(\alpha_1,\alpha_j)\\ \neq{} 0.
\end{align*}
The inverse function theorem then implies that $\Psi^{-1}$ is continuously 
differentiable at $\Psi(\theta_0)$, that is, for every $\varepsilon >0$, there
exists $\delta>0$ such that $\|\Psi(\theta_0) - \Psi_\infty(\theta)\|_\infty < \delta$
implies $\|\theta_0 - \theta\|_\infty < \varepsilon$. In particular, we obtain
 $$ \inf_{\|\theta_0 - \theta\|_\infty > \varepsilon} \|\Psi(\theta_0) - \Psi(\theta)\|_\infty \geq \delta,$$
that is, Eq.~\eqref{eq:uniftest1}, and the asymptotic normality and efficiency
of the posterior median for $\theta_0 \in \Theta'$ follow from Prop.~\ref{prop:uniftest}. 
Finally, we note that each extremal coefficient $\tau$ is continuously 
differentiable and strictly decreasing with respect to both components by Lemma
\ref{lem:dirichlet}. Thus, $\partial_{\alpha_1} \tau(\alpha_1,\alpha_2) < 0$
and $\partial_{\alpha_2} \tau(\alpha_1,\alpha_2) < 0$ for almost every
$\theta \in \Theta$.
\end{proof}

\subsection{Proofs from Subsection \ref{subsec:extremalt}}

\begin{proof}[Proof of Prop.~\ref{prop:extremalt}]
  By Prop.~4.3 in \cite{deo2017}, the model is differentiable in quadratic mean
  (even if $\nu > 0$ is not fixed). For any $1 \leq i_1 < i_2 \leq k$, and 
  fixed $\nu>0$, the pairwise extremal coefficient of the extremal-$t$ model
  with parameter matrix $\Sigma = \{\rho_{ij}\}_{1\leq i,j \leq k}$ is 
  \begin{align}\label{tau_ext}
    \tau_{i_1,i_2}(\Sigma) = 2 T_{\nu + 1}\left(\sqrt{(\nu+1) \frac{1 - \rho_{i_1i_2}}{1 + \rho_{i_1i_2}}} \right),
  \end{align}
  where $T_{\nu+1}$ denotes the distribution function of a $t$-distribution with
  $\nu+1$ degrees of freedom. Therefore, $\tau_{i_1,i_2}(\Sigma)$ as a function
  of $\rho_{i_1i_2}\in[-1,1]$ is strictly decreasing and the claim follows
  by Remark \ref{rem:mon} together with Prop.~\ref{prop:uniftest}. 
\end{proof}

\begin{proof}[Proof of Cor.~\ref{cor:ext}]
 Analogously to the proof of Cor.~4.4 in \cite{deo2017}, it can be shown that
 the model is differentiable in quadratic mean. Suppose that $\|t_1-t_2\|_2 
 \neq \|t_2-t_3\|_2$ and observe that the mapping
 $\Psi:\Theta \to \Psi(\Theta)$, $\theta = (s,\alpha) \mapsto \{\rho_{ij}\}_{1 \leq i, j \leq k}  = \{\exp(- \|t_i -t_j\|_2^\alpha/s)\}_{1 \leq i, j \leq k}$
 is continuously  differentiable. Since
 $$ \alpha = \frac{\log\{\log\rho_{12}\} - \log\{\log\rho_{23}\}}{\log \|t_1-t_2\|_2 - \log\|t_2-t_3\|_2},\qquad s= - \frac{\|t_1-t_2\|_2^\alpha}{\log\rho_{12}},$$
 the same holds true for the inverse mapping $\Psi^{-1}$. 
 
 Further, from the continuity of $\Psi^{-1}$ at $\Psi(\theta_0)$ at any 
 $\theta_0\in \Theta$, we obtain that for every $\varepsilon >0$ there is some
 $\delta>0$ such that for all $\Psi(\theta) \in \Psi(\Theta)$ with 
 $\|\Psi(\theta_0) - \Psi(\theta)\|_\infty < \delta$ we have 
 $\|\theta_0 - \theta\|_\infty < \varepsilon$. Consequently,
 $$ \inf_{\|\theta_0 - \theta\|_\infty > \varepsilon} \|\Psi(\theta_0) - \Psi(\theta)\|_\infty \geq \delta.$$
 From the proof of Prop.~\ref{prop:extremalt} and with the notation in 
 \eqref{tau_ext}, we obtain that
 $\|\Psi(\theta_0) - \Psi(\theta)\|_\infty \geq \delta$ implies
 $\max_{1\leq i_1<i_2\leq k} |\tau_{i_1,i_2}\{\Psi(\theta_0)\} - \tau_{i_1,i_2}\{\Psi(\theta)\}|>\delta'$ for some $\delta'>0$,
 that is, Equation \eqref{eq:uniftest1} holds. The assertion follows then from
 Prop.~\ref{prop:uniftest}.
\end{proof}

\subsection{Proofs from Subsection \ref{subsec:BR}}

\begin{proof}[Proof of Prop.~\ref{prop:cond-hr}]
  From Prop.~4.5 in \cite{deo2017}, it follows that the model is differentiable 
  in quadratic mean. For any $1 \leq i_1 < i_2 \leq k$, the pairwise extremal 
  coefficient of the H\"usler--Reiss model with parameter matrix 
  $\Lambda = \{\lambda^2_{i,j}\}_{1\leq i,j\leq k}$ is 
  \begin{align*}
    \tau_{i_1,i_2}(\Lambda) = 2\Phi_1\left\{ \sqrt{\lambda_{i_1,i_2}^2}\right\},
  \end{align*}
   which is a strictly increasing function in $\lambda_{i_1,i_2}^2 > 0$, and the claim follows
  by Remark \ref{rem:mon} together with Prop.~\ref{prop:uniftest}. 
\end{proof}

\begin{proof}[Proof of Cor.~\ref{cor:cond-br}]
 Analogously to Cor.~4.6 in \cite{deo2017}, the model can be shown to be 
 differentiable in quadratic mean. Suppose that $\|t_1-t_2\|_2 \neq 
 \|t_2-t_3\|_2$. As the mapping
 $\Psi:\Theta \to \Psi(\Theta)$, $\theta = (\lambda,\alpha) \mapsto \{\lambda^2_{ij}\}_{1 \leq i,j \leq k}  = \{ \frac {\|t_i - t_j\|_2^\alpha}{4 s}\}_{1 \leq i,j \leq k}$ 
 is continuously differentiable and
 $$ \alpha = \frac{\log\gamma_{12} - \log \gamma_{23}}{\log \|t_1-t_2\|_2 - \log \|t_2-t_3\|_2},\qquad s = \frac{\|t_i-t_j\|_2^\alpha}{4 \lambda^2 _{ij}},$$
 the inverse mapping $\Psi^{-1}$ is continuously differentiable, as well. The 
 same arguments as in the proof of Cor.~\ref{cor:ext} together with the proof 
 of Prop.~\ref{prop:cond-hr} yield the assertion.
\end{proof}

\bibliography{MaxStable_Bayesian}
\bibliographystyle{abbrvnat}

\end{document}